\documentclass{article}
\usepackage[utf8]{inputenc}
\usepackage{graphicx}

\usepackage{braket,amsfonts}
\usepackage{cancel}
\usepackage{graphicx}

\usepackage[a4paper, total={6in, 8in}]{geometry}

\usepackage{mathtools}

\usepackage{amsfonts}
\usepackage{dsfont}

\usepackage{amsmath}
\newcommand*\diff{\mathop{}\!\mathrm{d}}

\usepackage{appendix}

\usepackage{xpatch}

\usepackage{amsmath}
\DeclareMathOperator*{\argmax}{arg\,max}

\usepackage[unicode=true,
 bookmarks=true,bookmarksnumbered=true,bookmarksopen=true,bookmarksopenlevel=1,
 breaklinks=false,pdfborder={0 0 0},backref=false,colorlinks=true]
 {hyperref}
\hypersetup{
 linkcolor=blue, urlcolor=marineblue, citecolor=blue, pdfstartview={FitH}, unicode=true}

\usepackage{soul}
\usepackage{nicefrac}

\makeatletter
\newcommand*{\rom}[1]{\expandafter\@slowromancap\romannumeral #1@}
\makeatother

\usepackage{array}

\usepackage[caption=false]{subfig}

\usepackage{amsthm}

\newtheorem{remark}{Remark}[section]
\usepackage{subfig}
\usepackage{graphicx}
\usepackage{soul}

\usepackage{algorithmic}

\usepackage{graphicx,epstopdf}

\usepackage{amsopn}

\newtheorem{theorem}{Theorem}[section]

\usepackage{xspace}
\usepackage{bold-extra}
\usepackage[most]{tcolorbox}
\definecolor{OliveGreen}{rgb}{0,0.6,0}

\colorlet{texcscolor}{blue!50!black}
\colorlet{texemcolor}{red!70!black}
\colorlet{texpreamble}{red!70!black}
\colorlet{codebackground}{black!25!white!25}

\newcommand*{\dd}{\mathop{\kern0pt\mathrm{d}}\!{}}
\newcommand\restr[2]{{
  \left.\kern-\nulldelimiterspace 
  #1 
  \vphantom{\big|} 
  \right|_{#2} 
  }}

\lstdefinestyle{siamlatex}{%
  style=tcblatex,
  texcsstyle=*\color{texcscolor},
  texcsstyle=[2]\color{texemcolor},
  keywordstyle=[2]\color{texemcolor},
  moretexcs={cref,Cref,maketitle,mathcal,text,headers,email,url},
}

\tcbset{%
  colframe=black!75!white!75,
  coltitle=white,
  colback=codebackground, 
  colbacklower=white, 
  fonttitle=\bfseries,
  arc=0pt,outer arc=0pt,
  top=1pt,bottom=1pt,left=1mm,right=1mm,middle=1mm,boxsep=1mm,
  leftrule=0.3mm,rightrule=0.3mm,toprule=0.3mm,bottomrule=0.3mm,
  listing options={style=siamlatex}
}

\newtcblisting[use counter=example]{example}[2][]{%
  title={Example~\thetcbcounter: #2},#1}

\newtcbinputlisting[use counter=example]{\examplefile}[3][]{%
  title={Example~\thetcbcounter: #2},listing file={#3},#1}

\DeclareTotalTCBox{\code}{ v O{} }
{ 
  fontupper=\ttfamily\color{black},
  nobeforeafter,
  tcbox raise base, 
  colback=codebackground,colframe=white,
  top=0pt,bottom=0pt,left=0mm,right=0mm,
  leftrule=0pt,rightrule=0pt,toprule=0mm,bottomrule=0mm,
  boxsep=0.5mm,
  #2}{#1}

\title{Treatment-induced shrinking of {tumour aggregates}: A nonlinear volume-filling chemotactic approach}

\author{Luis Almeida\thanks{
Sorbonne Universit\'e, Inria, CNRS, Laboratoire Jacques-Louis Lions, F-75005 Paris. email: { luis.almeida@sorbonne-universite.fr}}
\and  Gissell Estrada-Rodriguez \thanks{Sorbonne Universit\'e, Laboratoire Jacques-Louis Lions, F-75005 Paris. email: {estradarodriguez@ljll.math.upmc.fr}}
\thanks{Corresponding author}
\and Lisa Oliver\thanks{Centre de Recherche en Cancérologie et Immunologie Nantes-Angers,
UMR 1232, Université de Nantes. email: { Lisa.Oliver@univ-nantes.fr}}
\and Diane Peurichard\thanks{
Sorbonne Universit\'e, Inria, CNRS, Laboratoire Jacques-Louis Lions, F-75005 Paris. email: { diane.a.peurichard@inria.fr}}
\and Alexandre Poulain \thanks{Sorbonne Universit\'e, Laboratoire Jacques-Louis Lions, F-75005 Paris. email: {poulain@ljll.math.upmc.fr}}
\and Francois Vallette\thanks{Centre de Recherche en Cancérologie et Immunologie Nantes-Angers,
UMR 1232, Université de Nantes. email: { Francois.Vallette@univ-nantes.fr}}}

\date{}

\begin{document}

\maketitle

\begin{abstract}
Motivated by experimental observations in 3D/organoid cultures derived from glioblastoma, 
we propose a novel mechano-transduction mechanism where the introduction of a chemotherapeutic treatment induces mechanical changes at the cell level. We analyse the influence of these individual mechanical changes on the properties of the aggregates obtained at the population level. We employ a nonlinear volume-filling chemotactic system of partial differential equations, where the elastic properties of the cells are taken into account through the so-called \emph{squeezing probability}, which depends on the concentration of the treatment in the extracellular microenvironment.  We explore two scenarios for the effect of the treatment: first, we suppose that the treatment acts only on the mechanical properties of the cells and, in the second one, we assume it also prevents cell proliferation. We perform a linear stability analysis which enables us to identify the ability of the system to create patterns and fully characterize their size. Moreover, we provide numerical simulations in 1D and 2D that illustrate the shrinking of the aggregates due to the presence of the treatment. 
\end{abstract}

\section{Introduction}

Glioblastomas (GBM) are solid tumours characterised by intra- and inter-tumoural heterogeneity and resistance to conventional treatments that result in a poor prognosis \cite{louis20162016}. They are the most common and aggressive primary brain tumour in adults. Standard treatments include surgical resection (when possible) combined with radiotherapy and chemotherapy using the DNA alkylating agent Temozolomide (TMZ) \cite{stupp2005radiotherapy}. In fact, the overall survival of treated patients is about 15 months versus 3 months without treatment, with fewer than 5\% of patients surviving longer than 5 years \cite{stupp2009effects}. 

One reason behind this relative therapeutic failure is the poor response of GBM tumours to this chemotherapeutic treatment, hypothesized to be due to their plasticity. Several studies have looked for the genetic compounds of TMZ-resistant cells focusing on the genes responsible for DNA mismatch repair protein \cite{talhaoui2017aberrant}, while other studies focused on  spatial and temporal variations in signalling pathways, which lead to functional and phenotypic changes in GBM \cite{lisaetal}.  Despite all these works, it remains difficult, from a biological and medical perspective, to investigate the connections between clinically observable glioma behaviour and the underlying molecular and cellular processes. The challenge is to integrate the theoretical and empirical acquired knowledge to better understand the mechanisms and factors that contribute to GBM resistance to treatment. In this context, mathematical models provide useful tools towards identifying links between different phenomena and how they are affected by the different therapeutic strategies. Much effort has been dedicated to the modelling of GBM  formation and invasion of the surrounding tissue, as well as to improving diagnosis and treatment. The exhaustive review \cite{alfonso2017biology} discusses  different modelling approaches as well as some of the main mechanisms that are observed in GBM formation and invasion. 

Recent studies highlighted the role of the communication between the tumour cells and the Tumor Micro-Environment (TME) and the properties of the Extra-Cellular Matrix (ECM) on tumour evolution and invasion \cite{di2018mining,brandao2019astrocytes,wurth2014cxcl12}. Cancer cells have been shown to respond to chemical and mechanical signals from components of the TME and \emph{vice versa}, and the interaction of tumour cells with the TME has been the subject of recent biological surveys \cite{hanahan2011hallmarks,fouad2017revisiting}. Many \emph{in vitro} (and \emph{ex vivo}) experiments have shown that cells that are cultured on ECM often have a tendency to form aggregate patterns that depend on the particular cell lines and physical properties of the media \cite{friedl2004prespecification}. 
{Biological evidence presented in \cite{maurer2019loss,kim2014cd44,chen2018feedforward} suggests that the formation of aggregates in glioma cells can be explained through a chemotaxis process (\emph{i.e.} the ability of cells to move along a chemical gradient), rather than, \emph{e.g.}, cell-cell adhesion. Chemokines, cytokines and growth factors involved in chemotaxis have been shown to be affected by the concentration of cells and therapies in most cancers and in particular in GBM \cite{roussos2011chemotaxis}.}

From a mathematical viewpoint, cell migration in the extracellular microenvironment  and the organisation of cells in response to chemical and mechanical cues are commonly studied using continuum descriptions based on differential equations  \cite{AGOSTI2020110203,painter2009modelling}. In a continuous setting, the chemotactic behaviour of cells is often modelled using a  Keller-Segel system of equations \cite{keller1971model}. This model was originally proposed for pattern formation in bacterial populations but turned out to be pertinent to describe a wide variety of self-organisation behaviours \cite{chaplain2006mathematical,ben2000cooperative,yamaguchi2005cell,kennedy1974pheromone,ward1973chemotaxis}.
Different variations of the Keller-Segel model have been adopted in order to better understand the way cells aggregate \cite{bubba2020discrete,alt1980biased,othmer2002diffusion,stevens2000derivation,dolak2005keller}. {Chemotactic-driven formation of aggregates of GBM was also proposed in \cite{aubert2006cellular} to reproduce experimental density profiles of
GBM spheroids. Biomechanical mechanisms, chemotaxis and cell-cell interaction have not been extensively studied in GBM. Nonetheless, recent results have shown that intra-cellular contacts in GBM are crucial not only for migration and growth but also for resistance to therapy \cite{weil2017tumor}. However, the mechano-transduction signals behind these phenomena have not been documented \cite{broders2018mechanotransduction}. In this work we hypothesized that aggregation observed during GBM treatment could be linked to physical/biomechanical phenomena related to cell-cell interactions and, in particular, we focus on the possible role of cell mechanical properties. }

{ Our hypotheses are based on recent results linking GBM biomechanics and resistance to treatment. Foss et al. \cite{foss2020patient} provided evidence that GBM heterogeneity could also be associated with mechanical phenotypes, since the physical properties of tumour tissue strongly influence aspects of tumour progression including cell cycle regulation, migration and therapeutic resistance. 
In \cite{hohmann2020macc1}, it was suggested that the over-expression of Metastasis Associated in Colon Cancer 1 (MACC1), which promotes cell motility, proliferation and metastasis in various cancers, increases the elastic modulus and migration and reduces adhesion of GBM cells. This occurs through increased amounts of protrusive actin on laminin which prevents 3D aggregate formation.}

In this paper, we follow the chemotactic approach to explain the formation of glioma aggregates, and we suppose that GBM cells react through a mechano-transduction mechanism to the presence of a drug. 
Inspired by experimental observations in 3D/organoid cultures derived from freshly operated GBM, which reproduce \emph{in vivo} behaviours as described in \cite{oizel2017efficient} (see Section \ref{sec: experiments} for more details), we explore a simple setting where GBM aggregate formation is due to nutrient-limited cell proliferation coupled with a chemotaxis-based cell movement. We propose a novel mechano-transduction approach where cells are able to change their individual mechanical properties in contact with the drug, and we study the influence of these individual mechanical changes on the characteristics of the aggregates obtained at the population level. 

We explore two scenarios: the case where the treatment only acts on the mechanical properties of the cells, and the case where it also prevents cell proliferation as was experimentally observed in \cite{lee2016temozolomide}. We adopt a macroscopic approach where cells are represented by their macroscopic density and are supposed to move in the environment via chemotaxis, \emph{i.e.} towards zones of high concentrations of a chemoattractant that is produced by the tumour cells. Moreover, cell proliferation is assumed to depend on the local concentration of nutrients available. {We remark that in this work, oxygen is considered just as a nutrient like the others. In future models, it will be interesting to be able to single out its particular roles.} Finally, we suppose that when the treatment - represented by its continuous concentration - is introduced, it diffuses in the environment and is naturally consumed by the cells.

Under these hypotheses, we obtain a nonlinear volume-filling chemotaxis model for the cell density, coupled with reaction diffusion equations for the chemoattractant and treatment concentrations. Moreover, we provide a linear stability analysis that enables us to study the ability of the system to generate patterns, and we provide numerical simulations in 1D and 2D.

The paper is organised as follows. In Section \ref{sec: experiments} we describe the \emph{in vitro} experiments and the main experimental observations that motivated our model.  Section \ref{sec: mathematical model} is devoted to the description of the model for the first part of the experiments (without the treatment, Section \ref{sec: first part}) and the second part, when the treatment is introduced (Section \ref{sec: part 2}). In Section \ref{sec: linear stability}, we present the stability analysis for each of these two parts, and Section \ref{sec: numerical simulations} is devoted to numerical simulations. Section \ref{sec: 1D numerical result} presents the results in 1D including a discussion on the comparison between numerical experiments and theoretical predictions of the stability analysis. Section \ref{sec: numerical 2d} shows the 2D simulations. Finally, we discuss the results and give some perspectives in Section \ref{sec: discussion and perspectives}.

\section{Experimental results of TMZ effect on glioblastoma 
}\label{sec: experiments}

To address the question of the response of GBM cells to TMZ treatment, we looked at recently developed 3D biosphere experiments, using GBM patient-derived cultures in a simple 3D scaffold composed of alginate and gelatin \cite{lisaunderreview}. GBMG5  cells were cultured at a concentration of $4\times 10^6$ cells/ml for $14$ days until the formation of cell aggregates could be observed, corresponding to the first part of the experiments, \textbf{P1}. Next, the cultures were treated with  $100\ \mu$M of TMZ for two hours once every week, which corresponds to the second part of the experiments, \textbf{P2}. Over the $30$ days, the proliferation was determined counting $3$ representative samples, and the cell number was determined as follows.  {The biospheres were dissociated by incubation for $3$ min in $100$ mM Na-Citrate and the cell number and cell viability were determined using the Countess optics and image automated cell counter (Life Technologies).}
In addition, the aggregates were photographed to analyse their morphology, and the diameter of the cell aggregates were measured from pictographs using FIJI. {To determine the diameters of these cell structures, more than 200 cell aggregates were measured.}

We show in Figure \ref{fig: data1} the mean length (in $\mu$m) of the cell aggregates computed from the microscopic images 
without TMZ treatment (round markers), and with TMZ weekly administered (squared markers).  Figures \ref{fig: data1} A) and B) show typical microscopic images of the spheroids at day 24, without and with weekly  TMZ treatment, respectively. In Figure \ref{fig: data2} we show the evolution of the cell number  as a function of time, {where we do not observe significant changes in cell number with and without TMZ, once the carrying capacity of the system is reached}.  {Based on these observations, we will suppose that TMZ acts as a non-cytotoxic drug, potentially inhibiting cell proliferation.} 

{Using clinically relevant concentrations of TMZ \cite{roos2007apoptosis}, the total number of cells in the biospheres does not seem to be significantly impacted by the TMZ treatment. However, the mean size of the GBMG5 cell aggregates decreases when TMZ is introduced weekly as compared to control cultures (Figure \ref{fig: data1}), suggesting that in the presence of the treatment, GBMG5 cells tend to self-organize into smaller and more compact cell clusters. }

\begin{figure}
    \centering
  \subfloat[]{\label{fig: data1}\includegraphics[scale=0.4]{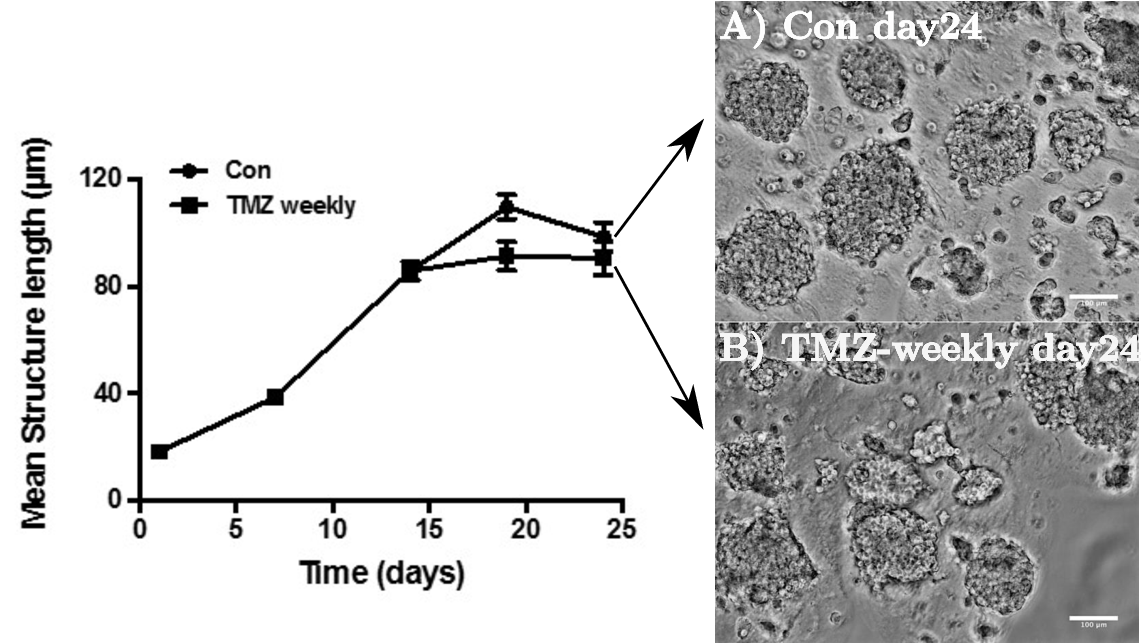}}
   \subfloat[]{\label{fig: data2}\includegraphics[scale=0.4]{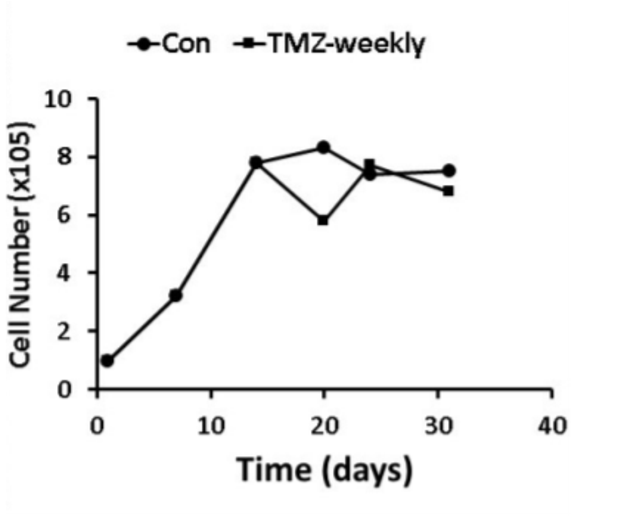}}

    \caption{Biological experiments of GBMG5 cells cultured in 3D scaffolds with and without {$100\mu$M} TMZ. (a) Evolution of the mean cell aggregates diameter determined from the microscopic images as function of time, without treatment (circle markers) and with a weekly TMZ (square markers). (A) Typical microscopic image on day 24 of control cell aggregates without treatment, (B) microscopic image at day 24 with  $100 \mu$M of TMZ administrated weekly for 2 hours. (b) Evolution of the total cell number in the biospheres as function of time, without treatment (circle markers) and with weekly TMZ (square markers). }
    \label{fig: fig1}
\end{figure}
\begin{figure}
    \centering
    \includegraphics[scale=0.55]{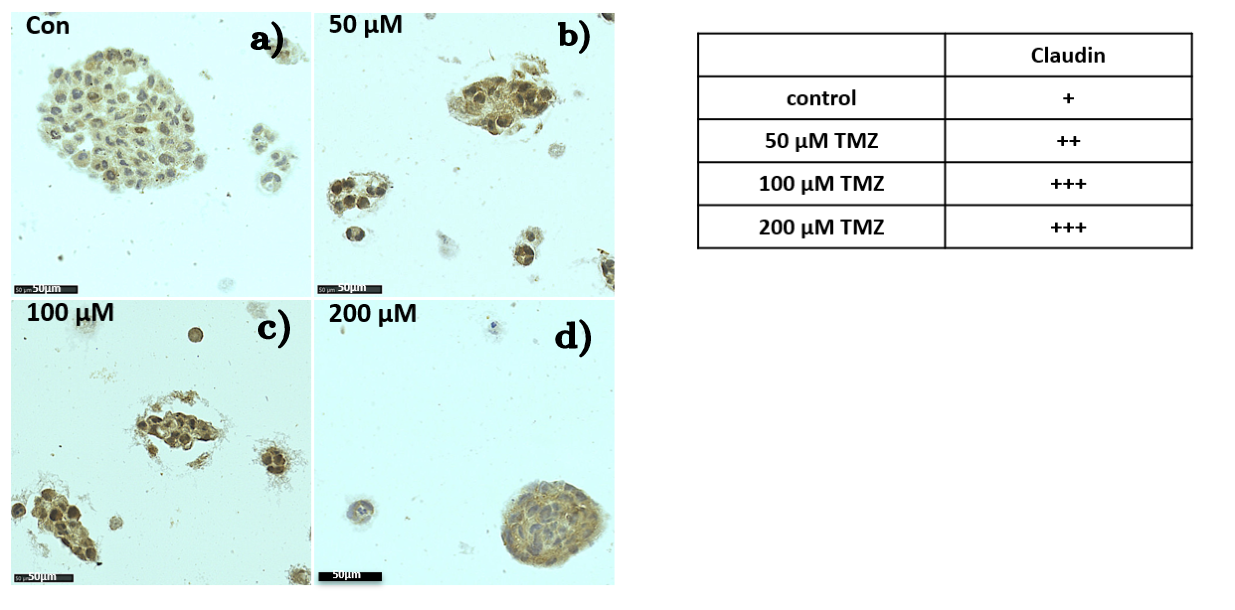}
    \caption{{Claudin expression marking the tight junctions between the cells.}}
    \label{fig: fig2}
\end{figure}

{Another observation supporting a tendency for GBM cells to form more compact structures with higher intracelluar adhesion under this type of treatment is the increased expression of claudin, a marker of tight junctions formed between cells. In Figure \ref{fig: fig2} we show GBMG5 cells that were cultured for 14 days until the formation of cell aggregates could be observed, and  the cultures were treated with different doses of TMZ until day 23. 
Furthermore, cell aggregates were fixed then labeled with an anti-claudin antibody. The degree of staining was determined in a double-blind experiment. From  Figure \ref{fig: fig2} (left), the cell aggregates seem to be smaller and more compact when TMZ is present, panel d), compared to the control group in a). These cultures are associated with higher levels of claudin (see Figure \ref{fig: fig2}, right). } 

Inspired by these observations and results presented in \cite{maurer2019loss,kim2014cd44,chen2018feedforward}, we propose a general model for the chemotaxis-driven formation of cell clusters and shrinking of the aggregates via the action of a non-cytotoxic drug. {We propose a novel approach based on a mechano-transduction mechanism, where GBM cells have the ability to change their mechanical properties in contact with the non-cytotoxic drug}. We also consider a second scenario where the treatment would also affect cell proliferation and compare the results obtained.

\section{Mathematical model}\label{sec: mathematical model}

Motivated by the experiments described in Section \ref{sec: experiments}, we assume that glioma cells have a chemotactic behaviour, \emph{i.e.} they move in response to some signaling chemical (chemoattractant), which is secreted by themselves and diffuses in the environment. The chemotactic movement of cell populations plays a fundamental role for example in gastrulation \cite{dormann2006chemotactic}, during embryonic development; it directs the movement of immune cells to sites of infection and it is crucial to understand tumour cell invasion \cite{condeelis2005great} and cancer development \cite{yamaguchi2005cell}. Motivated by these applications, chemotaxis and related phenomena have received
significant attention in the theoretical community, see the reviews \cite{hillen2009user,horstmann2004keller}.


We suppose that without treatment (part \textbf{P1} of the experiments described in Section \ref{sec: experiments}), {tumour} cells proliferate and move via chemotaxis as described before. We suppose that cell proliferation is limited by the nutrients available in the environment, \emph{i.e.} cell proliferation is only active as long as the local density does not exceed a given threshold corresponding to the carrying capacity of the environment. Moreover, in order to take into account the finite size of cells and volume limitations, cell motion is only allowed in locations where the local cell density is much smaller than another threshold value corresponding to the tight packing state. Without treatment, cells are supposed to behave as rigid bodies in the sense of \cite{painter2002volume}. In stressed conditions however, (\textbf{P2} of the experiments described in Section \ref{sec: experiments}) we suppose that cells respond to the chemotherapeutic stress, induced by the presence of {the treatment}, by changing their mechanical properties.

{These hypotheses are modelled via a system of partial differential equations (PDEs) which corresponds to a volume-filling chemotaxis equation \cite{painter2002volume}  for the first part (to describe the self-organisation of cells into aggregates), and an ``{semi-elastic}'' volume-filling chemotaxis approach \cite{wang2007classical} for the second part, when the treatment is introduced.}

{For convenience, we denote the density of the population of cancer cells in \textbf{P1}  by $u(\mathbf{x},t)$ and in \textbf{P2} by $w(\mathbf{x},t)$. Here $\mathbf{x}\in\rm{\Omega}\subset\mathbb{R}^2$ where $\rm{\Omega}$ is a bounded domain. The time $t\in[0,T]$, where $T=T_1+T_2$ represents the total time corresponding to the first and second parts of the experiments.}  {The main difference between these two populations is the change in the mechanical properties of the cells due to the presence of the treatment. If the concentration of the treatment is zero, $u(\mathbf{x},t)=w(\mathbf{x},t)$.} Cells follow a biased random walk according to the distribution of the chemoattractant of concentration $c(\mathbf{x},t)$ that is secreted by the cells. We start by detailing the different components of the mathematical models corresponding to \textbf{P1} and \textbf{P2} described in Section \ref{sec: experiments}.\\
\vspace{0.2cm}

\noindent\textbf{Logistic growth model for cell proliferation} { As previously described in the Introduction, in order to take into account the nutrient-limited population growth, cell proliferation is modelled by a logistic growth process. At the population level, we consider a source term $f(u)$ in the PDE for the evolution of the cell density, which depends nonlinearly on the local cell density $u$ and reads:}
\begin{equation}\label{eq: source term}
    f(u)=ru\Bigl(1-\frac{u}{u_\textnormal{max}} \Bigr)\ .
\end{equation}
Here, $r>0$ is the proliferation rate and $u_{\textnormal{max}}$ corresponds to the maximum density of the population, also referred to as the carrying capacity of the environment. Alternative cell kinetics could be considered. For example, we can assume that the proliferation is also mediated by the chemoattractant concentration such that $f(u,c)=ruc(1-u/u_{\textnormal{max}})$ \cite{painter2002volume}. {Here we only consider the case given by (\ref{eq: source term}).}\\

\noindent\textbf{Chemoattractant dynamics}
 {We suppose that cell aggregates spontaneously emerge as the result of a self-organisation phenomenon of chemotaxis type. To this aim, we suppose that the cells themselves produce the signaling chemical (chemoattractant) that drives their motion. The chemical secreted is therefore supposed to be continuously produced by the cells at rate $\alpha>0$ and diffuses in the surrounding environment with diffusion coefficient $d_2>0$. It is further assumed that the chemical has a finite lifetime and degrades at rate $\beta>0$. The evolution of the chemical concentration $c(\textbf{x},t)$ is therefore given by the following reaction-diffusion equation}
\begin{equation}
\partial_tc=d_2\Delta c+\alpha u -\beta c\ ,
\end{equation}
\noindent {where $u$ is the cell density.} \\

\noindent\textbf{Treatment dynamics}  We assume that the treatment is introduced at the beginning of \textbf{P2}. This treatment is supposed to diffuse in the environment with diffusion coefficient $d_4$, and to be consumed by the cells at rate $\delta$. This is modelled by a reaction-diffusion equation for the drug concentration $M(\textbf{x},t)$:
$$
\partial_t M=d_4\Delta M-\delta w\ ,
$$
where $w(\mathbf{x},t)$ represents the cell density corresponding to the second part of the experiments.
{Following \cite{baker1999absorption}, we suppose that the TMZ natural degradation rate \emph{in vitro} is very small compared to diffusion and absorption by the cells and neglect this term.}

We  consider different initial conditions for the drug: uniformly distributed in the simulation domain, introduced in the center as a very steep Gaussian function {or introduced through the boundary of the domain} (see Section \ref{sec: numerical simulations}).\\

\noindent\textbf{Cell mechanical properties}
{We introduce the so-called \emph{squeezing probability}, which describes the probability that a cell finds an empty space at a neighbouring location before moving {and it incorporates the cell-cell interactions \cite{wang2007classical}.}  The squeezing probability takes the form}
\begin{equation}
    q(u(\mathbf{x},t))=
    \begin{cases}
      1-\Bigl(\frac{u}{\bar{u}}\Bigr)^{\gamma(M)}, & \text{for}\ 0\leq u \leq\bar{u}\ , \\
      0, & \text{otherwise}\ ,
    \end{cases}\label{eq: squeezing probability}
  \end{equation}
   where $u(\mathbf{x},t)$ {is the cell density, $\gamma(M)\geq 1$ is the \emph{squeezing parameter}, $M \equiv M(\mathbf{x},t) \geq 0$ denotes the concentration of the treatment, and $\bar{u}$ is the \textit{crowding capacity} which corresponds to the tight packing state of the cells.}
 The function $q(u)$ satisfies the following properties,
\begin{equation}
    q(\bar{u})=0\ ,\ \ 0<q(u)\leq 1\ , \ \textnormal{and}\ \  q'(u)\leq 0\ \ \textnormal{for} \ \ 0\leq u<\bar{u}\ .
\end{equation}
Moreover, $|q'(u)|$ is bounded and $q''(u)\leq 0$.

The exponent $\gamma(M)$ is chosen to be
\begin{equation}
    \gamma(M)=(\bar{\gamma}-1)M+1 \ ,\label{eq: gamma M}
\end{equation} 
{where $\bar{\gamma}$ is a positive constant, $M=0$ when there is no drug in the environment (part \textbf{P1} of the experiments), and $M \equiv M(\mathbf{x},t)$  when the drug is introduced (part \textbf{P2}, described bellow). Such choice of $\gamma(M)$ enables to take into account different forms of the squeezing probability, corresponding to different mechanical behaviours of the cells. In Figure \ref{fig: sqeezing probability}, we plot the squeezing probabilities as a function of the cell density, corresponding to different values of $\gamma(M)$. We see that when there is no treatment present in the environment ($M=0$, $\gamma(M) = 1$, blue curve in Figure \ref{fig: sqeezing probability}), the squeezing probability decreases linearly with the local cell density, corresponding to cells modelled as solid particles.}

{However, for larger values of $\gamma(M)$ (when the treatment is present, $M>0$ and $\gamma(M)>1$, red and yellow curves in Figure \ref{fig: sqeezing probability}), the squeezing probability becomes a nonlinear function of the cell density, modelling the fact that in the presence of a drug, cells change their mechanical state to behave as {semi-elastic} entities that can squeeze into empty spaces. 
}

{We refer to \cite{wang2010chemotaxis,wang2007classical}  for more details on the link between the cells elastic properties and the squeezing probability. In particular, in \cite{wang2010chemotaxis}, the authors consider three different cell types leading to different squeezing probabilities:
\begin{itemize}
    \item If cells behave like a  fluid, they can fill all open spaces and cells interactions no longer have an impact on the squeeze probability, corresponding to the asymptotics $\gamma(M) \rightarrow \infty$.
    \item If cells are solid blocks, the squeezing probability is proportional to the number of occupants, \emph{i.e.} linearly dependent on the cell density, corresponding to the particular choice $\gamma(M) = 1$.
    \item If cells are elastic, they can adapt their configuration to squeeze into open spaces, leading to a nonlinear squeezing probability, piecewise higher than for solid blocks, corresponding to $\gamma(M) >1$. We will refer to this regime as semi-elastic.
\end{itemize}
}

\begin{figure}
    \centering
    \includegraphics[scale=0.5]{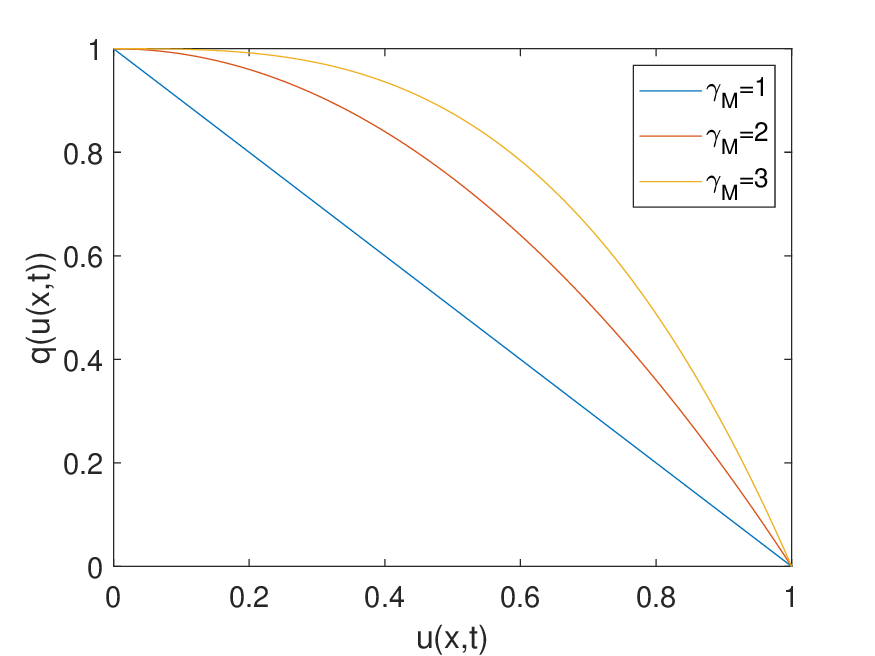}
    \caption{Squeezing probability for different values of $\gamma(M)$ for $\bar{u}=1$.}
    \label{fig: sqeezing probability}
\end{figure}

\subsection{Volume-filling approach for chemotaxis: first part \textbf{P1}}\label{sec: first part}

\noindent The classical Keller-Segel system of equations \cite{keller1971model} describes how cells {move along the gradient to local maxima of the chemoattractant.} At the same time, this chemical, which is produced by the cells, promotes aggregation leading to the so-called ``overcrowding scenarios'', and eventually, the cell density may blow up in finite time (see the comprehensive reviews \cite{hillen2009user,suzuki-chemotaxis-2018} and references therein).

{In this paper, in order to take into account volume limitations and the finite cell size, we consider a modified version of the Keller-Segel model called the \textit{volume-filling} approach for cell motion. This approach was introduced in \cite{painter2002volume}, where the authors provided a detailed derivation in one dimension as well as a comprehensive numerical study of the model.}

{For convenience of notation, for the case  without treatment \textbf{P1}, we define } 
\begin{equation}
    q_1(u(\mathbf{x},t))=1-\frac{u}{\bar{u}}\ \label{eq: q for part 1}. 
\end{equation}
{The Keller-Segel model built with this specific squeezing probability (\ref{eq: q for part 1}) has been widely studied in the literature,} from modelling \cite{painter2002volume,wang2010chemotaxis}, analytic \cite{wrzosek2010volume,han2017pattern,ma2012stationary,dolak2005keller} and numerical \cite{ibrahim2014efficacy} perspectives.\\

\noindent\textbf{Complete PDE system for the first part P1}
 {Taking into account all the previous ingredients, the complete PDE system for part \textbf{P1} (when no treatment is present in the environment) reads:}
\begin{equation}
\begin{aligned}
    \partial_t u & =\nabla\cdot(d_1D_1(u)\nabla u-\chi_u\phi_1(u)\nabla c)+f(u)\ ,\\ \partial_tc&=d_2{\Delta} c+\alpha u -\beta c\ ,\label{eq: formation of spheroids}
\end{aligned}
\end{equation}
{where the first equation describes the evolution of the cell density $u$ and the second is the reaction-diffusion equation for the chemoattractant, previously introduced. The equation for $u$ describes the volume-filling chemotactic motion associated with the squeezing probability $q_1(u)$. This equation has been obtained as the hydrodynamic limit of the continuous space-time biased random walk model that corresponds to the squeezing probability $q_1(u)$. In the hydrodynamic limit, the cell density evolves according to a nonlinear transport diffusion equation with source term, which corresponds to a volume filling Keller-Segel model with logistic growth. The density-dependent diffusion coefficient $D_1(u)$ and the chemotactic sensitivity $\phi_1(u)$ relate to $q_1(u)$ via
\begin{align*}
D_1(u) = q_1(u) - q_1'(u) u\ ,\ \ 
\phi_1(u) = q_1(u) u\ .
\end{align*}
{For \textbf{P1}, where $q_1(u)$ is given by (\ref{eq: q for part 1}), these coefficients take the values $D_1=1$ and $\phi_1(u)=u(1-u/\bar{u})$.}
In equation \eqref{eq: formation of spheroids},  $d_1, \chi_u, \alpha, \beta$ are all positive parameters and $f(u)$ is the logistic growth source term previously  defined in \eqref{eq: source term}.  
}

{The PDE system is supplemented with the following zero-flux boundary conditions}
\begin{equation}
\begin{aligned}
(d_1D_1(u)\nabla u-\chi_u\phi_1(u)\nabla c)\cdot\eta & = 0\ ,\ \ \
d_2\nabla c \cdot\eta =0\ ,\label{eq: boundary conditions}
\end{aligned}
\end{equation}
where $\eta$ is the outer unit normal at $\partial \rm{\Omega}$. The initial conditions are given by
\begin{equation}
\begin{aligned}
    u(\mathbf{x},0)&=u_0\ ,\ \ \
    c(\mathbf{x},0)=c_0\ .\label{eq: initial conditions}
\end{aligned}
\end{equation}


\subsection{PDE system including the treatment: Part \textbf{P2}}\label{sec: part 2}

\noindent {We now describe the dynamics of the cell population when the drug is introduced (part \textbf{P2} of the experiments described in Section \ref{sec: experiments}). As previously mentioned in Section \ref{sec: experiments} and motivated by the observations in \cite{lisaunderreview}, where the treatment TMZ does not seem to induce cell death, we suppose that the drug only affects the mechanical properties of the cells. For a cell density $w(\mathbf{x},t)$ the squeezing probability of part \textbf{P2} is}
\begin{equation}
    q_2(w(\mathbf{x},t), M)=1-\Bigl( \frac{w}{\bar{u}}\Bigr)^{\gamma(M)}\ .\label{eq: q for part 2}
\end{equation}
{Note that we have supposed that the crowding capacity $\bar{u}$, which corresponds to the tight packing state, remains unchanged from \textbf{P1} to \textbf{P2}. This corresponds to the hypothesis that the treatment does not modify the volume of the cells but only changes their elastic properties.}

The complete PDE system corresponding to the second part of the experiments therefore reads:
\begin{equation}
\begin{aligned}
    \partial_t w&=\nabla\cdot(d_3{D}_2(w,M)\nabla w-\chi_w{\phi}_2(w,M)\nabla c)+f(w)\ ,\\
    \partial_t c&=d_2\Delta c+\alpha w-\beta c\ ,\\
    \partial_t M&=d_4\Delta M-\delta w\ .\label{eq: second pahse}
\end{aligned}
\end{equation}
{Here, $f(w)$ is again the logistic growth source term given by  \eqref{eq: source term}, and the first equation has been derived using the squeezing probability $q_2(u)$ defined by \eqref{eq: q for part 2}. Note that the carrying capacity $u_{\textnormal{max}}$ remains unchanged in the two parts of the experiments: we have supposed here that the drug does not interact with the nutrients.} Analogous to (\ref{eq: formation of spheroids}), the movement of the cells is described by a chemotactic system where, in this case, the diffusion and chemosensitive coefficients depend also on the concentration of the treatment.  These modified coefficients will lead to changes in the size of the aggregates as shown in the numerical experiments in Section \ref{sec: numerical simulations}. The evolution of the concentration $c$ is the same as in (\ref{eq: formation of spheroids}) where, in this case, the chemoattractant is produced by the new population $w$. Including proliferation also in this second part allows us to assess the effect of the treatment in the population at earlier times, while the population of cells is still growing and aggregates are still forming.

{Different initial conditions for the cell density and concentration of the treatment will be considered as described in Section \ref{sec: numerical simulations}. Each initial condition for \textbf{P2} will correspond to a density profile solution of the system \textbf{P1} at a given time, \emph{i.e.} $w(\mathbf{x},0)=u(\mathbf{x},T_1)$ for some given $T_1$. We will consider cases where the treatment is introduced on already formed and stable aggregates (steady state of  \eqref{eq: final system 1}), but also cases where it is introduced at earlier times (during the formation of the aggregates, see Section \ref{sec: numerical simulations}). The initial condition for the drug concentration is considered to be either homogeneously distributed in the simulation domain, introduced in the center {or through the boundary}.}

\begin{remark}
In both parts of the experiments, we assume that the crowding capacity $\bar{u}$ is larger than the carrying capacity $u_\textnormal{max}$. 
{From a biological viewpoint, this hypothesis amounts to consider that the maximal local number of cells that can be supplied with enough nutrients/oxygen to survive, $u_{\max}$, is smaller than the purely mechanical density which corresponds to the optimal packing state of the cells. In a nutrient-unlimited environment cells would aggregate until maximally packed, while we suppose here that there is a limited amount of nutrients/oxygen in the environment which is not sufficient to support all the cells in a tight-packing state. }
\end{remark}

\section{Linear stability analysis and pattern formation}\label{sec: linear stability}
{The system  \eqref{eq: formation of spheroids} without source term is well known in the literature as the volume-filling Keller-Segel model, for which emergence of patterns has been characterised and is now well documented.} Pattern formation refers to the phenomenon by which, after varying a bifurcation parameter, the spatially homogeneous steady state loses stability and inhomogeneous solutions appear. In the following, we investigate in which parameter region we can expect instability of homogeneous solutions,  {corresponding to the formation of patterns}.
The linear stability analysis followed here {is classical and follows the lines of} \cite{painter2002volume,wang2007classical,murray2001mathematical}.

We first recall the two systems associated with the dynamics described in \textbf{P1} and \textbf{P2}. 
{Using the fact that in \textbf{P1} the squeezing probability is chosen to be $q_1(u)=1-\frac{u}{\bar{u}}$, we have}
\begin{equation}
\begin{aligned}
\begin{cases}
    \quad \partial_t u&=\nabla\cdot(d_1 \nabla u-\chi_u\phi_1(u)\nabla c)+ru\Bigl(1-\frac{u}{u_{\textnormal{max}}} \Bigr)\ ,\\ \quad \partial_t c &=d_2\Delta c+\alpha u-\beta c\ ,\label{eq: chemotaxis detail}
    \end{cases}
\end{aligned}
\end{equation}
where 
\begin{equation}
D_1=1\ ,\ \  \textnormal{and}\ \ \phi_1(u)=u \big(1-\frac{u}{\bar{u}} \big)\ .\label{eq: new diffusion and chemot coefficients} 
\end{equation}
For part \textbf{P2},  when the squeezing probability function is given by (\ref{eq: q for part 2}), the system writes
\begin{equation}
\begin{aligned}
\begin{cases}
    \quad \partial_t w&=\nabla\cdot(d_3{D}_2(w,M)\nabla w-\chi_w{\phi}_2(w,M)\nabla c)+\tilde{r}w\Bigl(1-\frac{w}{u_{\textnormal{max}}} \Bigr)\ ,\\
    \quad \partial_t c&=d_2\Delta c+\alpha w-\beta c\ ,\\
    \quad \partial_t M&=d_4\Delta M-\delta w\ ,
    \end{cases}
\end{aligned}
\end{equation}
with diffusion and chemotactic coefficients given by
\begin{equation}
{D}_2(w,M)=1+(\gamma(M)-1)\Bigl( \frac{w}{\bar{u}}\Bigr)^{\gamma(M)} \ \ \textnormal{and}\ \   {\phi}_2(w,M)=w\Bigl(1-\Bigl(\frac{w}{\bar{u}} \Bigr)^{\gamma(M)}\Bigr)\ .
\end{equation}



\subsection{Dimensionless model}\label{sec: dimensionless model}
{ To get a deeper insight of the system's behaviour we introduce the characteristic values of the physical quantities appearing in the models. Denoting by $X$ and $T$ the {macroscopic} units of space and time, respectively, such that}
$
\bar{\mathbf{x}}=\frac{\mathbf{x}}{X},\ \ \bar{t}=\frac{t}{T}\ ,
$ then we choose
\[ 
(\bar{\mathbf{x}},\bar{t})=\left(\sqrt{\frac{\beta}{d_2}}\mathbf{x},\ \frac{\beta d_1}{d_2}t \right)\ .
\]
{Using these new variables, the dimensionless system for part \textbf{P1} writes}
\begin{equation}
\begin{aligned}
\begin{cases}
\quad \partial_tu&=\nabla \cdot(\nabla u-A\phi_1(u)\nabla c)+r_0u\left(1-\frac{u}{u_{\textnormal{max}}}\right)\ ,\\
\quad \zeta\partial_t c&=\Delta c+u-c\ .\label{eq: final system 1}
\end{cases}
\end{aligned}
\end{equation}
{Similarly, we obtain for \textbf{P2}}
\begin{equation}
    \begin{aligned}
    \begin{cases}
\quad \theta\partial_tw&=\nabla\cdot({D}_2(w,M)\nabla w-B{\phi}_2(w,M)\nabla c)+\tilde{r}_0w\Bigl(1-\frac{w}{u_{\textnormal{max}}} \Bigr)\ ,\\
\quad \zeta\partial_t c&=\Delta c+w-c\ ,\\  \quad m\partial_tM& =\Delta M-\delta_0 w\ ,\label{eq: nondimensional second phase}
\end{cases}
\end{aligned}
\end{equation}
where
\begin{equation}
\begin{aligned}
    A=\frac{\chi_u}{d_1}\ ,\  r_0=\frac{d_2r}{d_1\beta}\ ,\  \zeta=\frac{d_1}{d_2}\ ,\  \theta=\frac{d_1}{d_3}\ ,\\  B=\frac{\chi_w}{d_3}\ ,\ \tilde{r}_0=\frac{d_2\tilde{r}}{d_3\beta}\ , \  m=\frac{d_3}{d_4}\ , \delta_0=\frac{\delta}{d_1} \ .\label{eq: nondimensional parameters}
\end{aligned}
\end{equation}

The parameters $\zeta$ and $m$ are assumed to be small since the chemoattractant and the chemotherapeutic treatment diffuse faster than the cells. On the other hand, $\theta\simeq 1$ since both population densities $u$ and $w$ are assumed to diffuse at similar rates.
{In the following, we state the linear stability for both systems in separate sections.}

\subsection{First part: Formation of the aggregates}\label{sec: first phase}
We first consider the system (\ref{eq: final system 1}), which can be re-written in a more general form as 
\begin{equation}
\begin{aligned}
\partial_t u&=\nabla\cdot(\nabla u -A\phi_1(u)\nabla c)+f(u)\ \label{eq: syst linear stability1},\\
\zeta\partial_t c  &= \Delta c+g(u,c)\ ,
\end{aligned}
\end{equation}
where $\phi_1(u)$ is given in (\ref{eq: new diffusion and chemot coefficients}), $f(u)=r_0u(1-\frac{u}{u_{\textnormal{max}}})$ and $g(u,c)= u- c$. This system is subject to uniformly distributed initial conditions and zero-flux boundary conditions as in (\ref{eq: boundary conditions}).

The main result in this section is the following theorem, which gives the  pattern formation conditions for the system (\ref{eq: syst linear stability1}). 

\begin{theorem}
Consider $(u^*,\ c^*) = (u_{\max},u_{\max})$ {the} spatially homogeneous steady state. Then pattern formation for the system (\ref{eq: syst linear stability1}) with zero flux boundary conditions (\ref{eq: boundary conditions}) is observed if the following conditions are satisfied,
\begin{equation}
    \begin{aligned}
    f_u^*+\zeta^{-1}g_c^*&<0\ ,\  f_u^*g_c^*>0\ ,\  \zeta^{-1}g^*_c+f^*_u-\zeta^{-1}g^*_uA\phi_1(u^*)>0\ , \\ &g_c^*+f_u^*+g_u^*A\phi_1(u^*)>2\sqrt{f_u^*g_c^*}\ .
    \end{aligned}
\end{equation}
The critical chemosensitivity is given by \begin{equation}
    A^c=\frac{2\sqrt{r_0}+ 1+r_0}{ u_{\textnormal{max}}\left(1-\frac{u_{\textnormal{max}}}{\bar{u}} \right)}\ ,
\end{equation}
and for $A>A^c$ patterns can be expected. The wavemodes $k^2$ are in the interval defined by \begin{align}
    k_1^2=&\frac{-m-\sqrt{m^2-4 f_u^*g_c^*}}{2}<k^2<k^2_2=\frac{-m+\sqrt{m^2-4 f_u^*g_c^* }}{2}\ ,
\end{align}
where $m=-(g_c^*+f_u^*+g_u^*A\phi_1(u^*))$.
\end{theorem}

\begin{proof}
See Appendix \ref{app: stability analysis} for the proof of this theorem.
\end{proof}




\subsection{Second part: Treatment}\label{sec: treatment}
We now consider the system (\ref{eq: nondimensional second phase}) which corresponds to \textbf{P2}, when the treatment is introduced. The parameter range where patterns are observed is summarised in the following theorem.

\begin{theorem}
Consider {${(w^*,\ c^*,\ M^*)=(u_\textnormal{max},\ u_\textnormal{max},\ M_s)}$, where  \newline$ M_s=|\rm{\Omega}|^{-1}\int_{\rm{\Omega}} M(\mathbf{x},0)\diff\mathbf{x}$}, a spatially homogeneous steady state. Also, consider (\ref{eq: nondimensional second phase}) with zero flux boundary conditions given by
\begin{align*}
    (d_3 D_2(w,M)\nabla w-\chi_w\phi_2(w,M)\nabla c)\cdot\eta & = 0\ ,\ \ \
    d_2\nabla c\cdot\eta = 0\ ,\ \ \
    d_4\nabla M\cdot\eta = 0. 
\end{align*}
Then, the critical chemosensitivity is given by 
\begin{equation*}
    B^c=\frac{2\sqrt{\bar{r}_0D_2(u_\textnormal{max},M_s)}+D_2(u_\textnormal{max},M_s)+\tilde{r}_0}{u_\textnormal{max}\Bigl(1-\Bigl(\frac{u_\textnormal{max}}{\bar{u}} \Bigr)^{\gamma_{M_s}} \Bigr)}\ ,
\end{equation*}
 where 
 \begin{equation*}
 D_2(u_\textnormal{max},M_s)=1+(\gamma_{M_s}-1)\Bigl(\frac{u_\textnormal{max}}{\bar{u}} \Bigr)^{\gamma_{M_s}}\ .
 \end{equation*}
 Patterns can be expected if $B>B^c$ and the wavemodes $k^2$ are in the interval defined by
 \begin{align*}
    k_1^2=&\frac{-\bar{m}-\sqrt{\bar{m}^2-4D_2(u_\textnormal{max},M_s)(f_w^*g_c^*)}}{2D_2(u_\textnormal{max},M_s)}<k^2<k^2_2\\ & \ \ \ \ \ \ \ \ \ \  =\frac{-\bar{m}+\sqrt{\bar{m}^2-4D_2(u_\textnormal{max},M_s)(f_w^*g_c^*)}}{2D_2(u_\textnormal{max},M_s)}\ ,
\end{align*}
for $\bar{m}=-(D_2(u_\textnormal{max},M_s)g_c^*+g_w^*B\phi_2(u_\textnormal{max},M_s)+f^*_w)$.
\end{theorem}

\begin{proof}
The proof of this theorem can also be found in  Appendix \ref{app: stability analysis}.
\end{proof}

\begin{remark}
For the case of $2$ dimensions, we can rewrite the systems (\ref{eq: final system 1}) and (\ref{eq: nondimensional second phase}) using  polar coordinates $(\rho,\theta)$ where we use the transformation $x=\rho\sin\theta$, $y=\rho\cos\theta$ and the Laplace operator  is now given by $\Delta_p=\frac{1}{R}\frac{\partial}{\partial \rho}\left(\rho\frac{\partial}{\partial \rho}\right)+\rho^2\frac{\partial^2}{\partial\theta^2}$, where $R$ is the radius of the domain. The eigenvalue problem (\ref{eq: psi}) is now written as $-\Delta_p\psi_k=k^2\psi_k$ with boundary conditions $\partial\psi_k/\partial \rho=0$ at $\rho=R$. 
The eigenfunctions are obtained  by separation of variables and are given by $\psi_k(x,y)=\mathcal{R}(\rho)\Theta(\theta)$. Here  $\Theta(\theta)=e^{is\theta}=A\cos(s\theta)+B\sin(s\theta)$ for some $s\in\mathbb{Z}$. The radial part $\mathcal{R}(\rho)$ is given in terms of Bessel functions $\mathcal{R}(\rho)=\mathcal{J}_s(k\rho)$ (see \cite{sarfaraz2018domain}) where $k=\frac{c_{s,p}}{R}$ and $c_{s,p}$ denotes the $p$th zero derivative of $\mathcal{J}_s$, which is a first kind Bessel function of order $m$. Finally we can write
$\psi_k^{s,p}(\rho,\theta)=\mathcal{J}\left( \frac{c_{s,p}}{{R}}\rho\right)(A\cos(s\theta)+B\sin(s\theta))\ .$ 
\end{remark}

{The stability analysis reveals that several competing effects control the system's ability to create patterns (aggregates). The criteria obtained both in \textbf{P1} or \textbf{P2} show that the chemotactic sensitivity must be large enough to compensate the smoothing effect of the diffusion term and of the logistic growth. On the other hand, one can observe from the bifurcation formulae that the ratio $\frac{u_{\max}}{\bar{u}}$ (carrying capacity \emph{vs.} density of the tight packing state) plays an important role in the emergence of patterns: larger values lead to more aggregated states. These results show that the logistic growth term has an intrinsic smoothing property, \emph{i.e.}, it tends to force the density to equate the carrying capacity, while the chemotactic term acts as an attractive force and creates zones of higher density (recall that $u_{\max}<\bar{u}$). The aggregates are an expression of a balance in between these two competing effects, which are completely characterised by the stability criterion.  }

\begin{figure}[tbhp]
    \centering
  \subfloat[]{\label{fig: wavenumber}\includegraphics[scale=0.4]{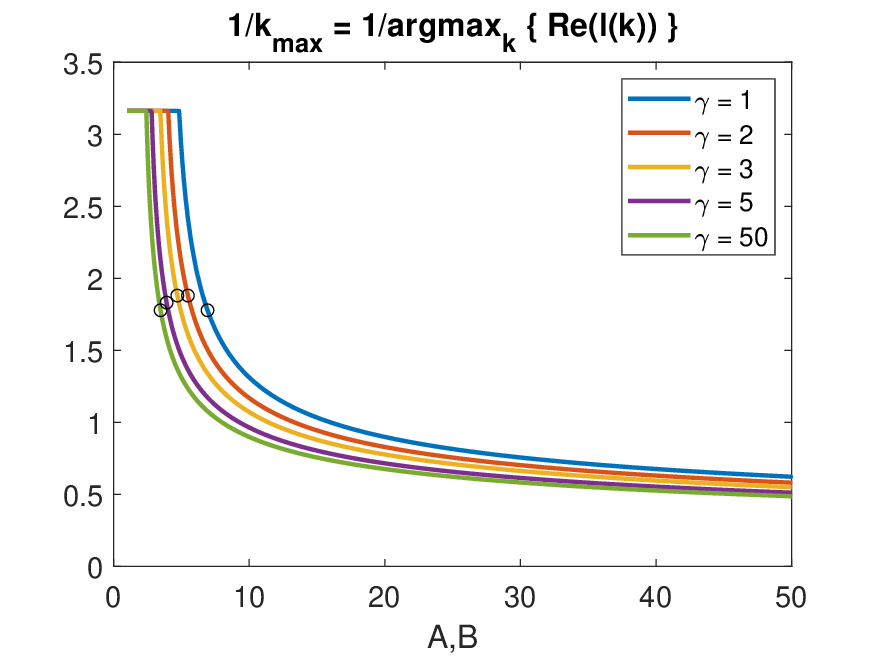}}
   \subfloat[]{\label{fig: wavenumber 2D}\includegraphics[scale=0.4]{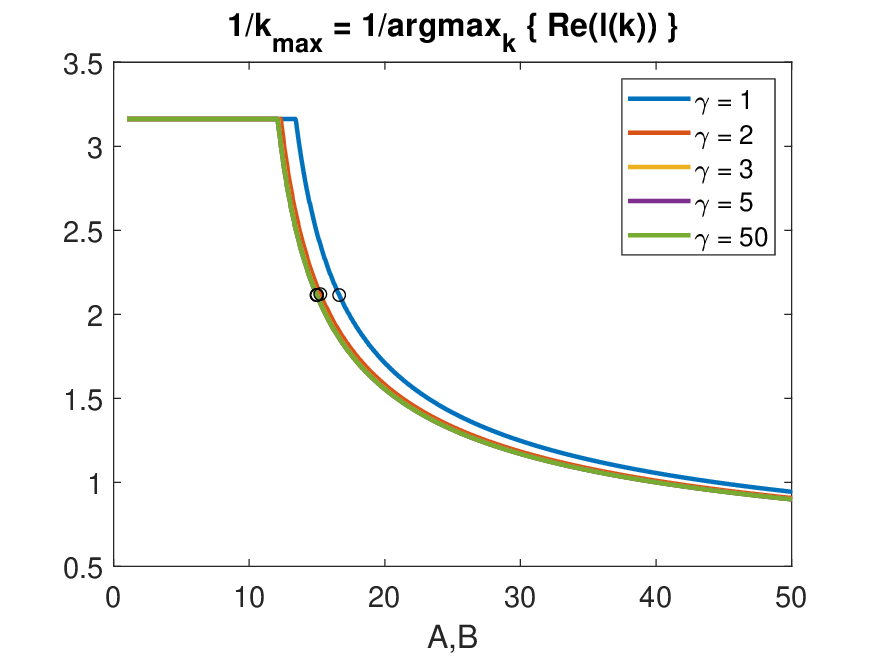}}
    \caption{Wavenumbers for different values of $\gamma(M)$ when (a) $r_0=0.1$, $u_{\max} = u_0=0.5$; and for (b) $r_0=0.05$, $u_{\max} = u_0=0.1$ . Circles indicate $k_c$ values.}
    \label{fig: maximum eigenvalues and modes}
\end{figure}
{In order to give more insights about the size of the emerging patterns, we show in Figure~\ref{fig: maximum eigenvalues and modes} {the values of the inverse maximal wavenumber} as function of the chemotactic sensitivity (denoted by $A$ for part \textbf{P1} and $B$ for part \textbf{P2}), for different values of the exponent $\gamma(M)$. We recall that $\gamma(M) = 1$ in \textbf{P1}, where cells act as rigid bodies, and $\gamma(M)>1$ in the presence of the treatment, where cells behave as {semi-elastic entities} (\textbf{P2}). Figure \ref{fig: wavenumber} shows the results for growth rate $r_0=0.1 = 0.8 \; \rm{day}^{-1}$ and carrying capacity $u_{\max} = 0.5$, Figure \ref{fig: wavenumber 2D} for  $r_0=0.05 = 0.4 \; \rm{day}^{-1}$ (slower growth) and $u_{\max} = 0.1$ (lower carrying capacity). In both figures, the black circles indicate the critical values for the chemotactic sensitivity above which the system is unstable. Here, the tight packing density is set to $\bar{u} = 1$. The maximal wavenumber corresponds to the most unstable mode, \emph{i.e.} the perturbed wave that will grow the fastest. Therefore, the inverse of this maximal wavenumber is directly related to the size of the emerging patterns. 

As one can observe, an increase in the chemotactic sensitivity parameter correlates with a decrease in the observed pattern size, suggesting that the aggregates are smaller: larger chemotactic sensitivity leads to more aggregated clusters. Moreover, the aggregate size also decreases as cells pass from rigid bodies to {semi-elastic} entities (when $\gamma(M)$ increases). This is due to the fact that for larger values of $\gamma(M)$, cells are more easily deformed and can aggregate more efficiently than when they behave as rigid spheres. }

{When we increase the ratio $\frac{u_{\max}}{\bar{u}}$ (compare Figure \ref{fig: wavenumber} and \ref{fig: wavenumber 2D}), we observe that the critical value of the chemosensitivity above which patterns are generated is larger than for smaller ratios $\frac{u_{\max}}{\bar{u}}$. These results highlight once again the smoothing effect of the logistic growth: when the cell tight packing density is unchanged, decreasing the carrying capacity of the environment enhances cell death in the aggregates formed by chemotaxis, where cells  try to reach the tight packing state. In this case, the critical chemosensitivity value must be large enough to compensate for the cell death induced by the logistic growth. Moreover, we observe that larger ratios $\frac{u_{\max}}{\bar{u}}$ induce less influence of the function $\gamma(M)$. The cell aggregation abilities are mainly driven by the chemosensitivity intensity and not so much by the cell mechanical properties for large values of $\frac{u_{\max}}{\bar{u}}$.}
\section{Numerical simulations}\label{sec: numerical simulations}
In addition to the analytic results obtained in Section \ref{sec: linear stability}, we present numerical simulations for the two problems \eqref{eq: final system 1} and \eqref{eq: nondimensional second phase}. This allows us to investigate {the behaviour of the  model's solutions} for different scenarios and range of parameters. It is well-known that a standard discretisation of the Keller-Segel models can lead to nonphysical solutions due to the convective term.  
Here, we focus on a numerical method that preserves the non-negativity of the cell density using the upwind finite element method described in \cite{bubba2019positivity} for the simulation of the Cahn-Hilliard equation. 

The calculation of the chemotactic coefficient follows the lines of \cite{almeida_energy_2019}. Indeed, the finite volume scheme proposed in \cite{almeida_energy_2019} is identical to the numerical method presented in this paper in dimension one.   
{However, in higher dimensions, and since we also use a finite element method, the numerical scheme presented in \cite{almeida_energy_2019} differs from the one in this section, as detailed in Appendix \ref{app: numerics description}.  }

 \subsection{Biological relevance of the model parameters}
{
 Here, we comment on the choice of the model parameters that we will use for the numerical simulations and how they relate to experimental known data. As hypoxia-inducible factors (HIF) are supposed to be responsible for the chemotaxis motion of GBM cells, we suppose that the diffusion coefficient $d_2$ and consumption rate $\beta$ for the chemoattractant are linked to biological measurements of the oxygen diffusion in human brain which were estimated in \cite{Toma2013,Frieboes2007} to  $d_2 = 86.4\ \textnormal{mm}^2\ \textnormal{day}^{-1}$ and $\beta = 8640\ \textnormal{day}^{-1}$. {For such values, and using the scaling of Section \ref{sec: dimensionless model}, one unit of time of our model corresponds to $0.125$ days, and one unit of space is $10^{-1}$ mm.} In \cite{Frieboes2007,Columbo2015}, the proliferation rate for well oxygenated glioma cells \emph{in vitro} $r$, was shown to lie between $0.5$ and $1$ day$^{-1}$. As the proliferation rate relies significantly on the nutrient, also smaller value seems to be biologically admissible in real conditions and following \cite{Columbo2015}, we choose $r=0.4\ \textnormal{day}^{-1}$ and $r = 0.8\ \textnormal{day}^{-1}$ (corresponding to the non-dimensionalised parameter $r_0 = 0.05$ and $r_0=0.1$).  
 {These values are in agreement with the non-dimensional parameters in \eqref{eq: nondimensional parameters}.} 
 
 As we found no experimental data on the chemotactic coefficient $\chi_u$ of glioma cells in response to chemoattractant concentration, the choice of this parameter is driven by the stability analysis and we find that the interesting regimes are obtained for a dimensionless chemosensitivity in between $7$ and $70$, corresponding to a chemotactic coefficient $\chi_u \in [0.6, 6]\ \textnormal{mm}^2\  \textnormal{day}^{-1}$. Moreover, as no measurements for glioma cells' diffusion coefficient are available in the literature, the parameter $d_1$ is arbitrarily chosen to be $100$ times smaller than the chemoattratant diffusion speed and we choose $d_1 \approx d_3 \approx 0.086\ \textnormal{mm}^2\ \textnormal{day}^{-1}$, \emph{i.e} the non-dimensionalised parameter $\zeta = \frac{d_1}{d_2} = 0.01$.    }

\subsection{Numerical results for a one dimensional case}\label{sec: 1D numerical result}

For all numerical computations we choose the packing capacity $\bar{u}=1$. We consider different proliferation rates $r_0=0.1,\ 0.05$ and different initial conditions and carrying capacities $u_{\max} = u_0 =0.1,\ 0.5$.
 The nondimensional parameters given in (\ref{eq: nondimensional parameters}) are $\zeta=m=0.01$ and $\theta=1$ since we assume that the chemoattractant $c$ and the treatment diffuse much faster than the cells, while the motility of the cells is not affected by the treatment, so $d_1\approx d_3$. 
The initial condition for the cell density $u_0$ is assumed to be randomly distributed in space. We similarly define the initial chemoattractant concentration $c_0$.  

In this section we start by solving the systems (\ref{eq: final system 1}) and (\ref{eq: nondimensional second phase}) on the interval $[0,L]$ with homogeneous non-flux boundary conditions using the method described in Appendix \ref{app: numerics description}. In Appendix \ref{app: one dimension numerics} we investigate the effect of the size of the domain  as well as the effect of the parameters $A$ and $B$ on the formation and evolution of patterns. Moreover, using (\ref{eq: nondimensional second phase}) we study the effect of the treatment using the solution of (\ref{eq: final system 1}) at the final time $T_1$ as initial condition. We explore the case when we introduce the treatment at earlier stages of the formation of the aggregates. 

We consider two different scenarios for the evolution of the concentration of the treatment. First, we assume that the treatment diffuses very fast in the whole domain so that the concentration is homogeneous from time $T_2=0$. The other two cases we consider, which are closer to real experiments, start with a high concentration of the drug, either in the centre or at the boundary of the domain, and this concentration diffuses over time according to the third equation in (\ref{eq: nondimensional second phase}).\\


\noindent\textbf{Comparison with the linear stability analysis} In order to quantify the aggregate sizes and compare it to the ones predicted by the stability analysis, we use the Fourier transform of the numerical solution and extract the frequency that corresponds to the maximal Fourier mode. For the sake of this analysis, periodic boundary conditions are therefore considered. To this aim, we compute the discrete Fourier transform $\mathcal{F}[u](\mathbf{x},t)=\hat{u}(\lambda,t)$  and define
$$k_\textnormal{max} = \argmax_\lambda(|\hat{u}(\lambda)|) =\argmax_\lambda \Bigl(\sqrt{\textnormal{Re}(\hat{u}(\lambda))^2+\textnormal{Im}(\hat{u}(\lambda))^2 }\Bigr)\ ,$$
{which corresponds to the frequency of the largest Fourier mode. The inverse ${(k_{\max})}^{-1}$ of this maximal frequency relates to the pattern size. This maximal frequency of the Fourier transform of the solution is expected to correspond to the maximal wavenumber predicted by the stability analysis. We show in Figure \ref{fig: comparison analytic vs numeric} the values of ${(k_{\max})}^{-1}$ computed from the numerical solution (blue dotted line) compared to the predictions of the stability analysis (red curve), as function of the chemosensitivity, for $\gamma(M) = 1$ (Figure \ref{fig: comparison A}) and $\gamma(M) = 5$ (Figure \ref{fig: comparison B}). As one can observe, we obtain a very good agreement between the numerical values and the predictions of the stability analysis, and we recover the critical value of the chemosensitivity parameter above which the system generates patterns.  }\\

\begin{figure}[tbhp]
\centering
\subfloat[]{\label{fig: comparison A}\includegraphics[scale=0.4]{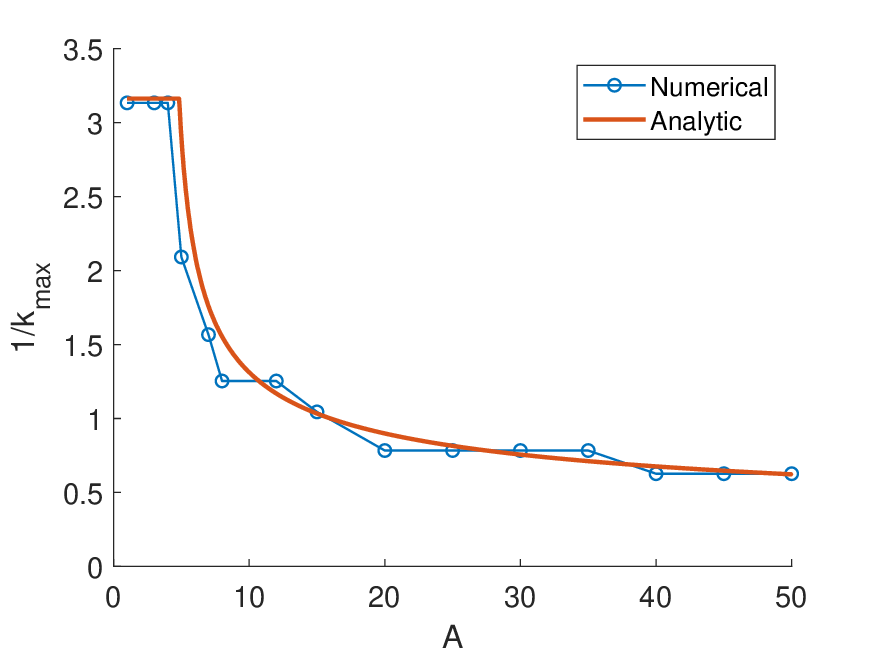}}
\subfloat[]{\label{fig: comparison  B}\includegraphics[scale=0.4]{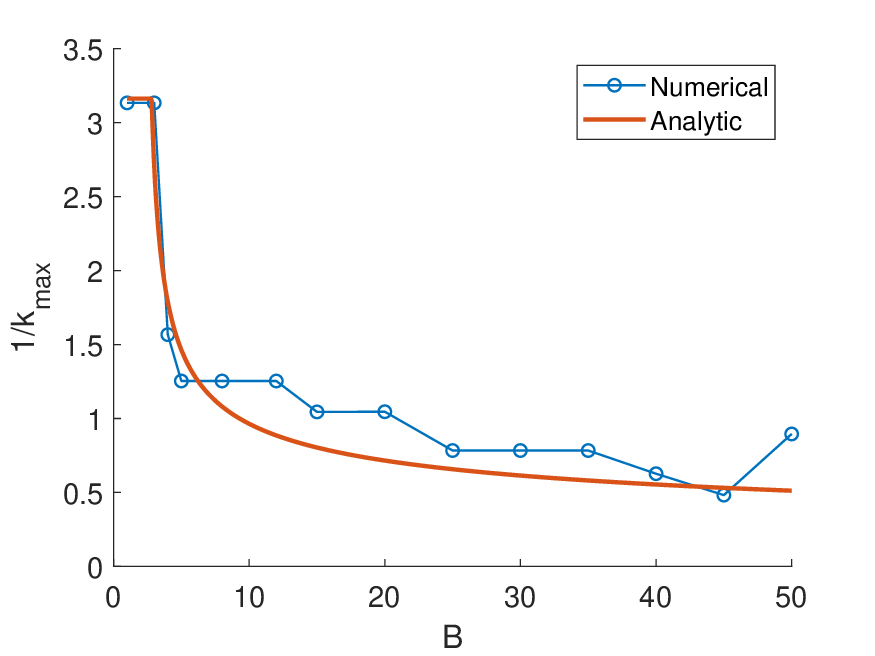}}
\caption{Comparison of the wavelength obtained analytically, using (\ref{eq: critical chemosensitivity}) and (\ref{eq: k critic 1}), and numerically, using the Fourier transform of the solution for (a) $\gamma(M)=1$ and (b) $\gamma(M)=5$.}\label{fig: comparison analytic vs numeric}
\end{figure}

\noindent\textbf{Introduction of the treatment on already-formed aggregates}
 In this part, we aim to study the influence of the treatment on already formed aggregates. For this, we let the system run in \textbf{P1} (without treatment, $M=0,\ \gamma(M) = 1$) until time $T_1=25$ days, and introduce the treatment uniformly in the domain ($M=1$, $\gamma(M) = 5$). 

In Figure \ref{fig: comparison different AB} left and middle panels, we choose values of the chemosensitivity very close to the critical values corresponding to $k_c$, where the wavenumbers are very different for the cases $\gamma(M)=1$ and $\gamma(M)=5$ as we see in Figure \ref{fig: wavenumber}. In Figure \ref{fig: comparison different AB}, the blue curves describe the formation of aggregates for a time $T_1=25$ days without the treatment, while the cells are proliferating with rate $r_0=0.1 = 0.8 \; \rm{day}^{-1}$. We consider two different scenarios when introducing the treatment: either cells stop proliferating (red curves), or they continue with the same rate as before $r_0=0.8 \; \rm{day}^{-1}$ (yellow curves).


When we introduce the treatment for values of $A$ and $B$ close to the instability threshold ($A=B=7$, Figure \ref{fig: comparison different AB} left), we observe that the aggregates become steeper and the density in each aggregate reaches the packing capacity $\bar{u}=1$. This clearly leads to more compact aggregates as a result of the nonlinearity introduced in  (\ref{eq: second pahse}) by  the function $\gamma(M)$. The main physical difference between changing the function $\gamma(M)$ and changing the chemosensitivity coefficients $A$ or $B$ is the following. {By changing $A$ or $B$ depending on the concentration of the treatment, we are enhancing aggregation over diffusion, essentially we are changing the motility of cells. By changing $\gamma(M)$ the motility, as well as the elastic properties of the cells in the aggregates are affected.}   When we introduce the treatment while cells keep proliferating, aggregates tend to merge together since the density is growing, {as we can observe in the middle panel of Figure \ref{fig: comparison different AB} by comparing the solution without treatment (blue curve) and with TMZ drug (yellow curve)}.


{It is noteworthy that for large values of the chemosensitivity parameter ($A=B=70$, Figure \ref{fig: comparison different AB} right panel), the treatment does not impact the aggregate dynamics. In this case, cell aggregation is mainly driven by the chemotactic term and the cell mechanical properties have little influence. These observations are in agreement with the stability analysis, which shows that the function $\gamma(M)$ has more influence when the chemotactic sensitivity is close to the instability threshold.}

\begin{figure}[tbhp]
\hspace{-1.1cm}\includegraphics[scale=0.4]{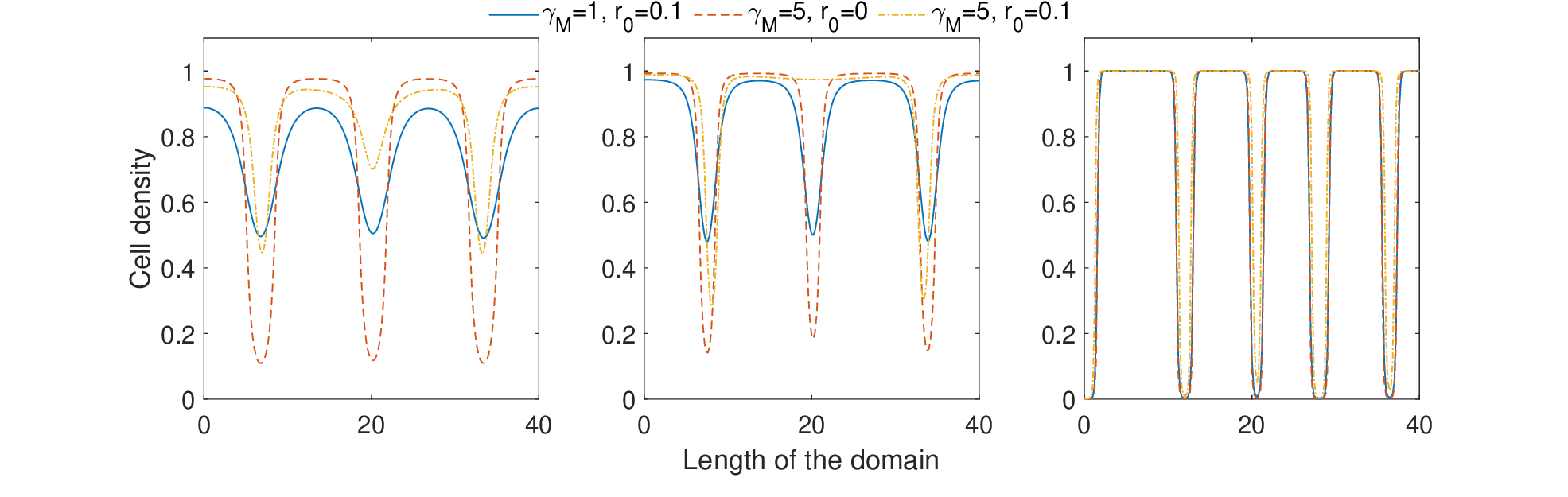}
    \caption{Aggregation pattern when $A=B$ for $A=7$ (left), $A=12$ (middle) and $A=70$ (right). The blue curves are solutions of \textbf{P1} at $T_1=25$ days. The red and yellow curves are solutions of \textbf{P2} at $T=T_1+25$ days when the treatment is introduced uniformly and $r_0=0$ or $r_0=0.1=0.8 \; \rm{day}^{-1}$ respectively.
    }
    \label{fig: comparison different AB}
\end{figure}

\begin{figure}[tbhp]
    \centering
    \includegraphics[scale=0.4]{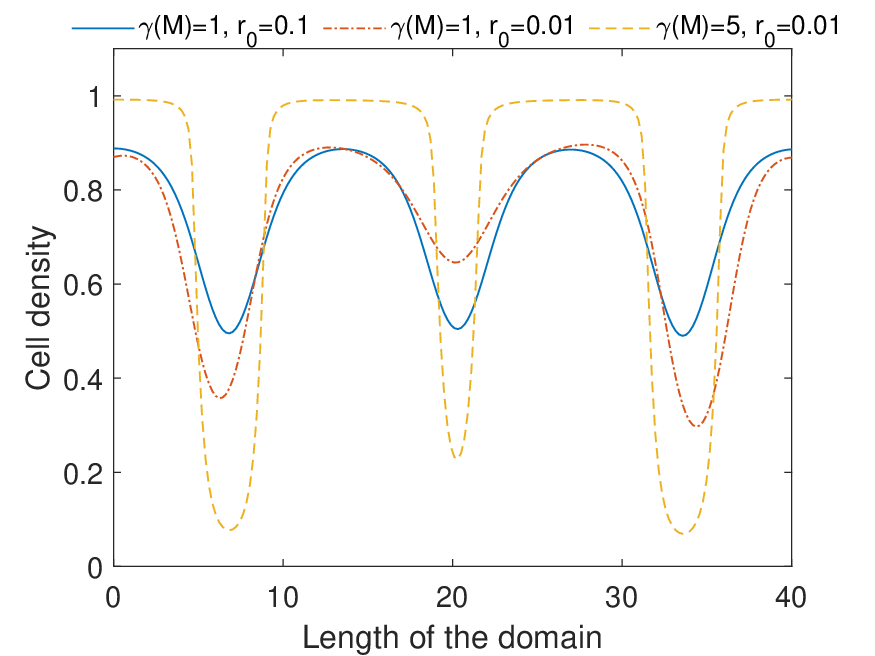}
    \caption{Solution for $A=B=7$. The blue curves is the solution of \textbf{P1} at $T_1=25$ days. The red and yellow curves are the solutions of \textbf{P2} at $T=T_1+25$ days with reduced proliferation rate for $\gamma(M)=1$ and $\gamma(M)=5$, respectively.}
    \label{comparison no elastic variation}
\end{figure}

{Finally, by comparing the red and yellow curves in Figure \ref{fig: comparison different AB}, it is clear that cell proliferation has a major impact on the size of the aggregates. If the treatment has the double effect to stop proliferation as well as modifying the cell mechanical properties, cell aggregates will become very dense and well-separated, while merging aggregates are still observed if the treatment has the sole effect to change the cell mechanical properties.}

{In Figure \ref{comparison no elastic variation} we compared the results when the TMZ only affects the proliferation rate and the mechanical properties of the cells remain unchanged (red curve), with the case when TMZ affect both, proliferation and elasticity (yellow curve). We observe that when TMZ only reduces the proliferation rate of GBM cells, we do not observe a significant sharpening of the cell clusters, compared to the case when TMZ has the coupled action of reducing proliferation and changing the mechanical properties of the cells. Together with Figure \ref{fig: comparison different AB}, we therefore conclude that the mechano-transduction phenomenon induced by the presence of the drug could be sufficient to explain the shrinking of the cell aggregates documented in Section \ref{sec: experiments}.} 

In all previous experiments we chose to introduce the treatment at time $T_1=25$ days, we now aim to study the effect of the treatment introduced at different times in the aggregation process. \\

\noindent\textbf{Introduction of the treatment at different times}
{Here, we consider the case when the treatment is introduced at different times in the aggregation process. As before, we let the system run in \textbf{P1} (without drug, $M=0,\gamma(M) = 1$) until time $t=T_1$ and introduce the treatment uniformly in the domain ($M=1$, $\gamma(M) = 5$). We consider the cases when the treatment has the ability to stop proliferation, and when the treatment only acts on the cell mechanical properties.  In Figure \ref{fig: different IC}, we show the results at time $T=T_1 + 25$ days (red curve), when the treatment is introduced at times $T_1=6.25$ days (left plots), $T_1 = 12.5$ days (middle plots) and $T_1 = 37.5$ days (right plots). Figures \ref{fig: comparison IC} shows the results when the treatment stops proliferation while Figures \ref{fig: comparison IC prol} shows the results when the treatment only changes the mechanical properties of the cells. For each, the blue curves are the density profiles before introducing the treatment.  

As one can observe, in Figures \ref{fig: comparison IC} and \ref{fig: comparison IC prol}, introducing the treatment at different times of the aggregation process {has} a major impact on the size of the aggregated patterns formed at a latter time. Introducing the treatment at an earlier time ($T_1=6.25$ days, left figures) enables to obtain smaller aggregates compared to when the treatment is introduced on already formed aggregates ($T_1=37.5$ days, right figures). This effect is more visible when the treatment has the double effect of blocking cell proliferation and changing the elasticity (compare red curves in Figures  \ref{fig: comparison IC} and \ref{fig: comparison IC prol}). In this case, the earlier the treatment is introduced, the smaller the aggregated patterns. When the treatment stops proliferation as well as it changes the cell mechanical properties, and is introduced at later times (right panel of Figure \ref{fig: comparison IC}) we recover the observation of the real systems, where the treatment induces a shrinking of the aggregate and favors the formation of more compact cell structures. This effect is not observed when proliferation is active with the treatment, (right panel of Figure \ref{fig: comparison IC prol}) where aggregates are merging and they are larger than before the treatment introduction. This suggests indeed that the treatment has the double effect of blocking cell proliferation as well as changing the cell mechanical properties. The model suggests that introducing the treatment at earlier times of the tumour development could enable to control the size and separation of the tumour aggregates. }\\

\begin{figure}[tbhp]
    \hspace{-0.6cm}\subfloat[]{\label{fig: comparison IC}\includegraphics[scale=0.4]{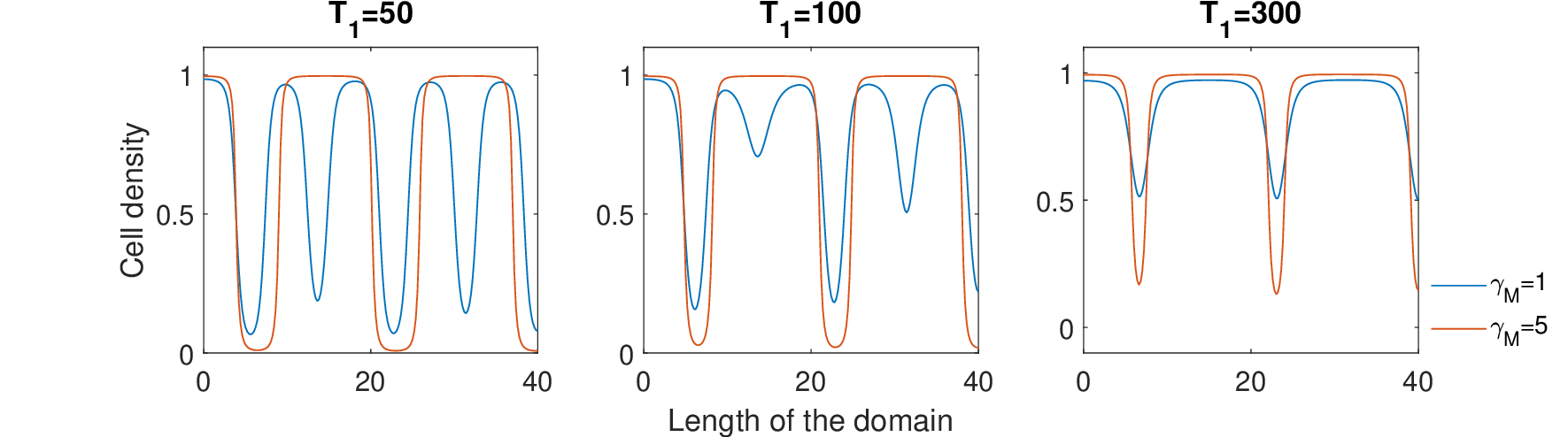}}
    
    \hspace{-0.6cm}\subfloat[]{\label{fig: comparison IC prol}\includegraphics[scale=0.4]{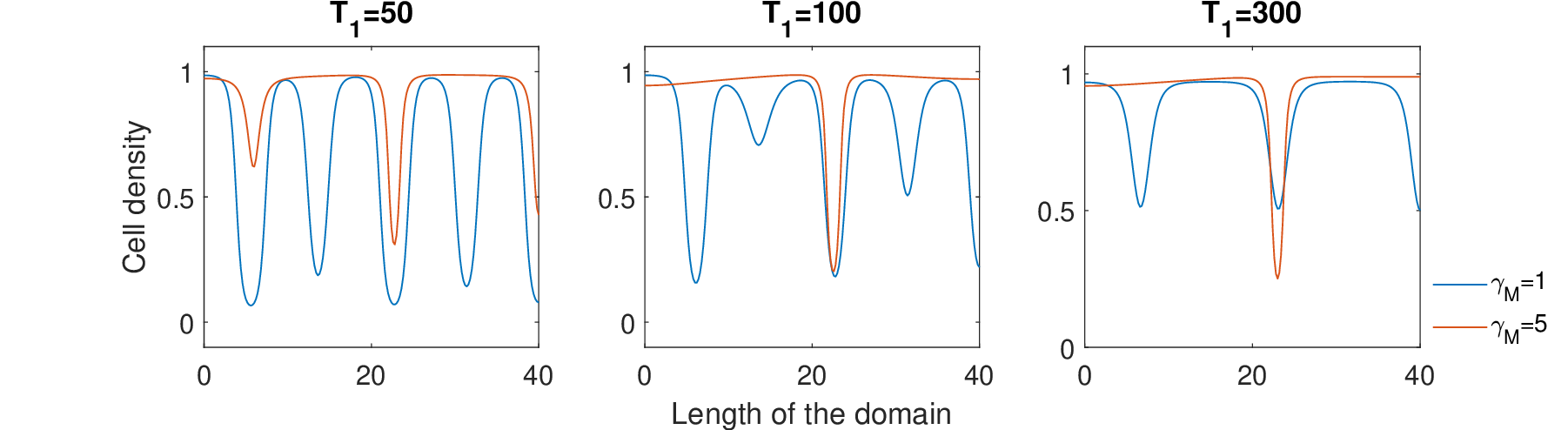}}
    \caption{Introduction of  the treatment at different times $T_1=6.25,\ 12.5,\ 37.5$ days when $u_0=0.5$ and $A=B=12$. The blue curves give the initial condition $u(\mathbf{x},T_1)$ for the part \textbf{P2}. In (a) $r_0=0$ and in (b) $r_0=0.1=0.8 \; \rm{day}^{-1}$. The red curves are at $T=T_1+25$ days.}
    \label{fig: different IC}
\end{figure}

\noindent\textbf{Introduction of the treatment in the middle of the domain}
\begin{figure}
    \hspace{-0.3cm}\includegraphics[scale=0.4]{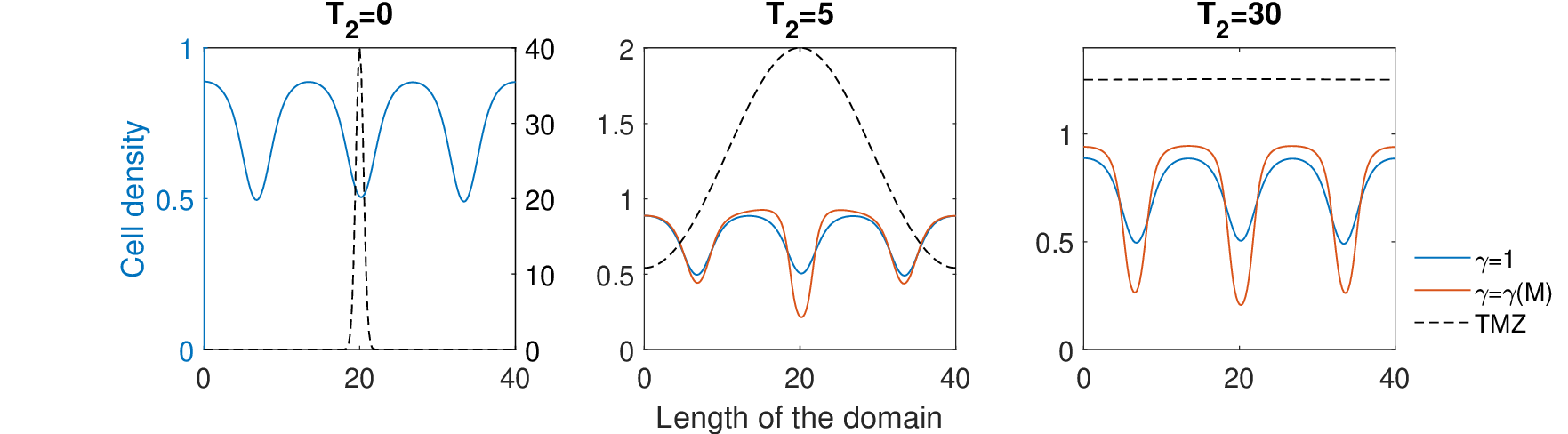}
    \caption{Introduction of the treatment in the centre of the domain. The blue curves correspond to the solution of \textbf{P1} at $T_1=25$ days and the red curves are the solution of \textbf{P2} at $T=T_1+0.6$ days and $T=T_1+3.75$ days.}
    \label{fig: center TMZ 1D}
\end{figure}
Finally, we aim to study the case when the treatment is introduced in the center of the domain and diffuses in the environment. Here we assume that the treatment is not consumed  or escapes the domain, therefore $\delta_0=0$ in (\ref{eq: nondimensional second phase}). In Figure \ref{fig: center TMZ 1D}, we show the density profiles of the solution before introducing the treatment (blue curves), and when the treatment is present (red curves), at times $T_2=0$ (left), $T_2=0.6$ days (center) and $T_2=3.75$ days (right). The distribution of the treatment follows a Gaussian of the form $M(\mathbf{x},0)=Ce^{\frac{(\mathbf{x}-\mathbf{x}_0)^2}{2\sigma^2}}$, where $C=40$ is the amplitude, $\mathbf{x}_0$ is the center of the Gaussian and $\sigma=0.5$ describes the spread. As one can observe, the large concentration of the treatment in the middle immediately sharpens the interface between already-formed aggregates, and favors the separation of the cell clusters. {As the treatment diffuses in the domain (see the middle figure of Fig. \ref{fig: center TMZ 1D}), the cell cluster interfaces sharpen, creating denser and well-separated cell clusters.} 


\subsection{Numerical results for a two dimensional case}\label{sec: numerical 2d}
For the 2D simulations we consider that $\rm{\Omega}$ is a disk of radius $R$ which can be defined as $\rm{\Omega}=\{ (x,y)\in \mathbb{R}^2:\ x^2+y^2<R^2\}$ where the boundary is given by $\partial \rm{\Omega}=\{(x,y)\in\mathbb{R}^2:\ x^2+y^2=R^2 \}$.
{The proliferation rate is chosen to be $r_0=0.4 \; \rm{day}^{-1}$ and the initial homogeneous density as well as the carrying capacity are set to $u_{\max} = u_0=0.1$ or $u_{\max} = u_0=0.5$.  The other parameters can be found at the beginning of Section \ref{sec: 1D numerical result}.}

In Figure \ref{fig: formation of spheroids 2D} we show the formation of aggregates for different values of $A$ without the treatment, for $u_0=0.1$ and $r_0=0.4 \; \rm{day}^{-1}$. We observe that for $A=10$ (Figure \ref{fig: 2D without TMZ A10}) we do not have patterns, in agreement with the analytic results obtained in Figure \ref{fig: comparison IC prol}  since this value of $A$ is less than $A^c \approx 16.7$. As we increase the chemosensitivity parameter, the aggregates become more compact. From Figure \ref{fig: 2D without TMZ A20 } we observe the phenomena of two aggregates merging together, analogous to the one dimensional results in Figure \ref{fig: formation spheroids 01}.
As expected, by changing {the carrying capacity and the initial density of cells to $u_0 = u_\textnormal{max} = 0.5$}, the patterns change shape. We observe a transition from spot-like patterns in Figure \ref{fig: formation of spheroids 2D} to maze-like structures in Figure \ref{fig: formation of spheroids 2D diff IC}. This behaviour has been widely studied experimentally \cite{ouyang1991transition}, numerically \cite{meinhardt1989models,nagorcka1992stripes} and more recently, also including a volume-filling approach \cite{painter2002volume}.

\begin{figure}[tbhp]
  \centering
  \subfloat[]{\label{fig: 2D without TMZ A10}\includegraphics[scale=0.28]{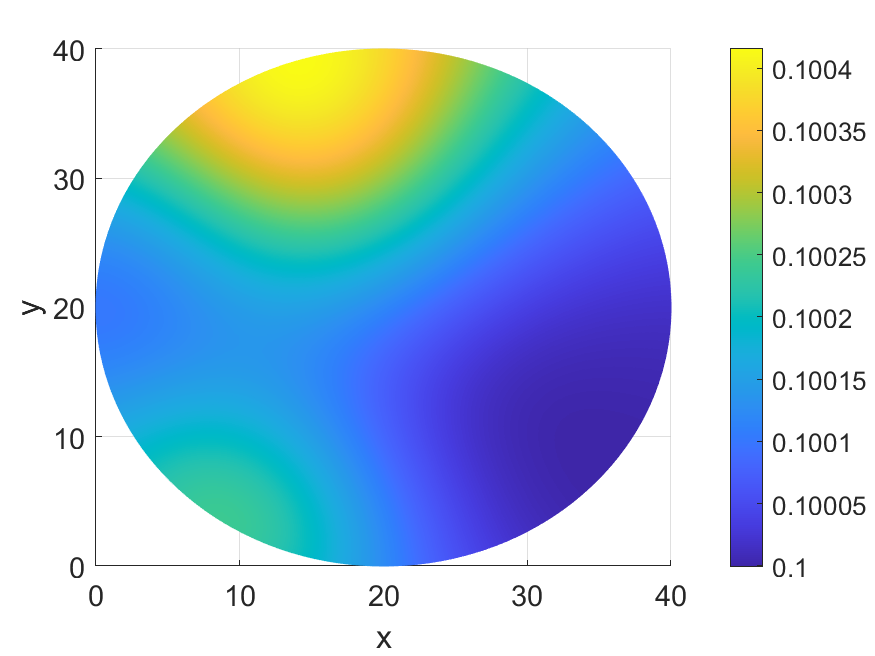}}
    \subfloat[]{\label{fig: 2D without TMZ A20 }\includegraphics[scale=0.28]{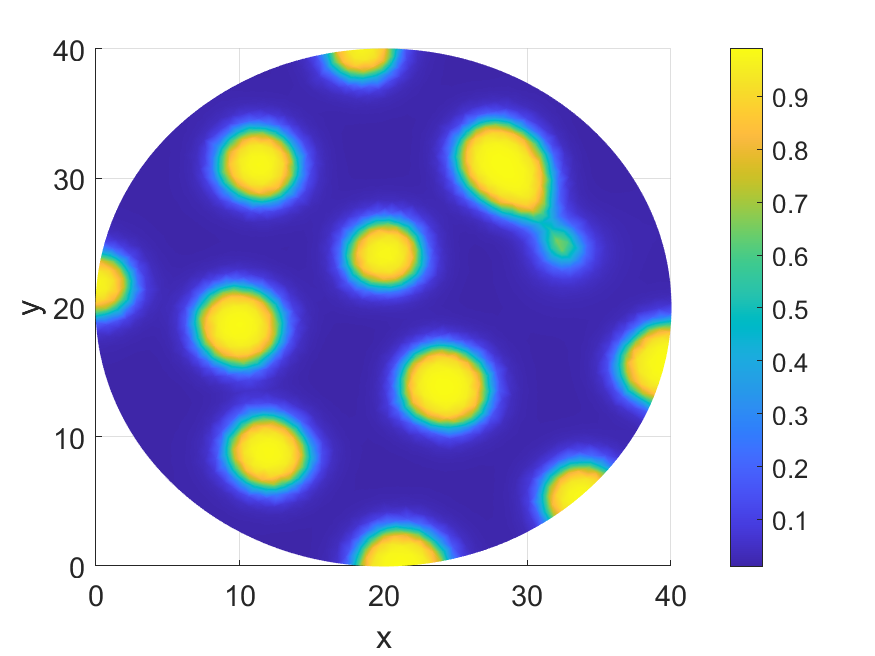}}
    \subfloat[]{\label{fig: 2D without TMZ A70}\includegraphics[scale=0.28]{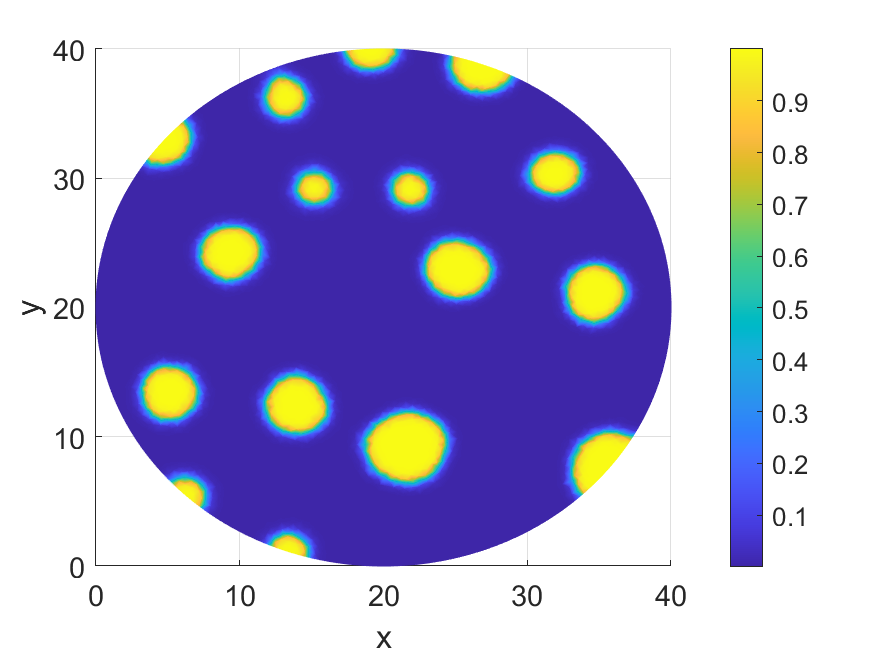}}
    \caption{Formation of aggregates at $T_1=25$ days when $u_0=0.1$, $r_0=0.4 \; \rm{day}^{-1}$ and (a) $A=10$, (b) $A=20$ and (c) $A=70$.}\label{fig: formation of spheroids 2D}

  \subfloat[]{\label{fig: 2D without TMZ A7 }\includegraphics[scale=0.28]{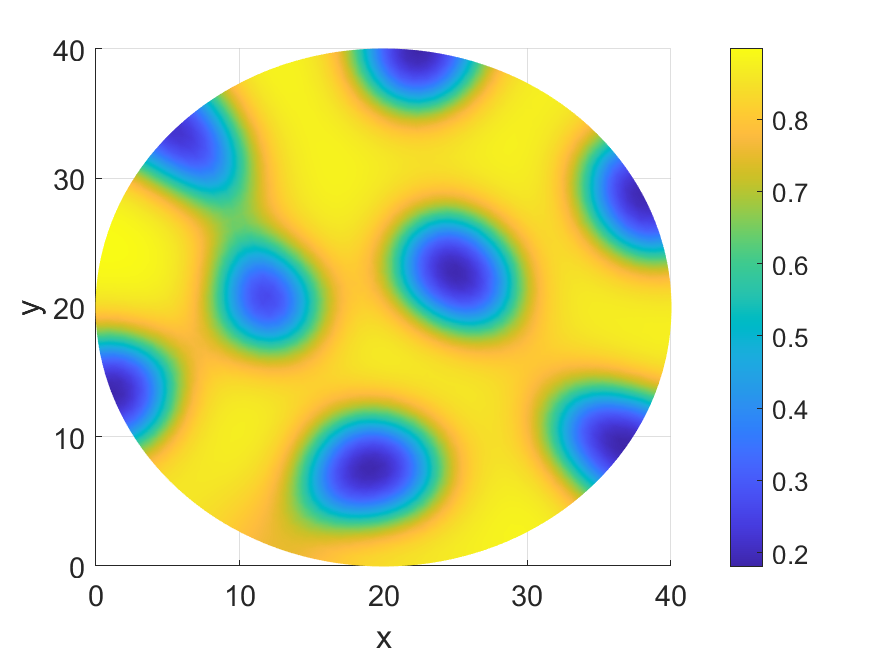}}
    \subfloat[]{\label{fig: 2D without TMZ A12 }\includegraphics[scale=0.28]{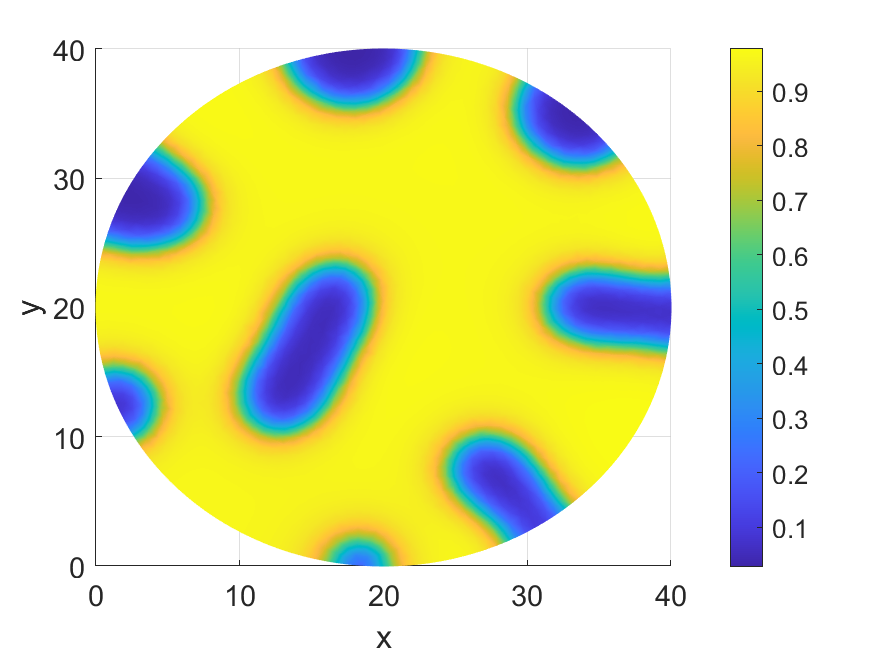}}
    \subfloat[]{\label{fig: 2D without TMZ A50 }\includegraphics[scale=0.28]{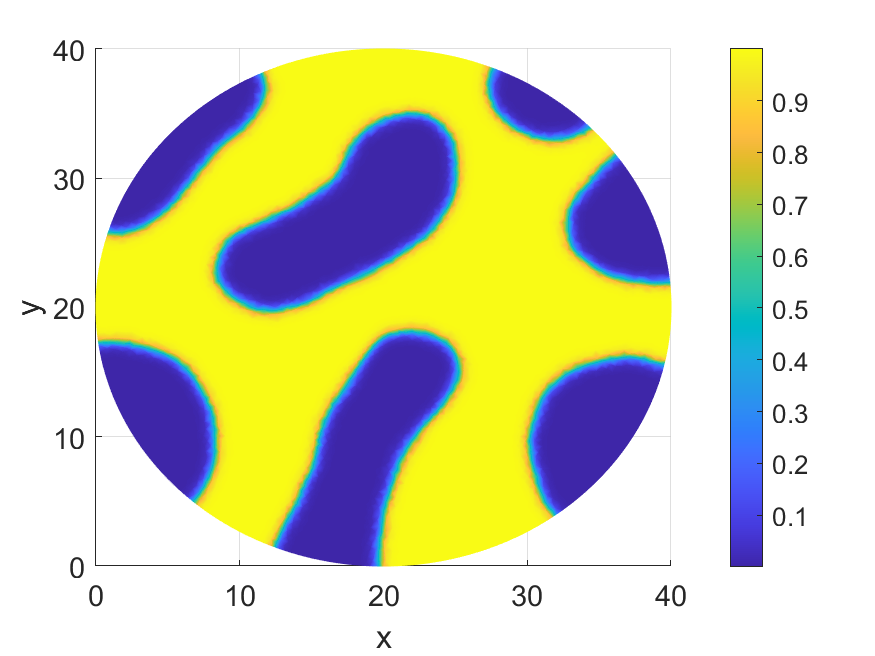}}
    \caption{Formation of aggregates at $T_1=25$ days without the treatment and when $u_\textnormal{max} = u_0=0.5$, $r_0=0.4 \; \rm{day}^{-1}$ and (a) $A=7$, (b) $A=12$ and (c) $A=50$.}\label{fig: formation of spheroids 2D diff IC}
\end{figure}

{Analogous to the one dimensional case, we consider two different initial conditions for the treatment: (i) we first include the treatment uniformly in the domain with $M=5$, and (ii) we introduce the treatment with a localised concentration either in the middle of the domain or at the boundary, and let it {diffuse} in space. In these simulations, the treatment is supposed to block proliferation as well as changing the cell mechanical properties.} 

{In order to compare the change in size of the aggregates before and after the treatment, we compute the difference between the solution of the first part of the experiments $u(\mathbf{x},T_1)$ coming from (\ref{eq: final system 1}), when the aggregates are formed (at $T_1=25$ days), and the solution  $w(\mathbf{x},T_2)$ of (\ref{eq: nondimensional second phase}), once the treatment has been inserted (at time $T = 25 + T_2$ days). In Figure \ref{fig: difference solutions}, we show the results for different values of the chemosensitivity parameter $A=B$, when the aggregates have been exposed to a uniform distribution of the treatment for a time $T_2=3.75$ days. {The positive values of this difference, observed on the external boundary of the clusters, indicate the shrinking of the aggregate due to the action of TMZ, while negative values on the inner boundaries of the clusters highlight the increased concentration on the boundaries. TMZ therefore favors the production of smaller and more concentrated aggregates, as described in Section \ref{sec: mathematical model}}. When we compare the results in Figure \ref{fig: difference solutions 100} and \ref{fig: difference solutions 1000} for different values of $A$ and $B$ we observe that the effect of the treatment is stronger when the value of the chemosensitivity parameter $B$ is closer to its critical value (see Figure \ref{fig: wavenumber 2D}). This is in accordance with the results of the stability analysis. The cell mechanical properties (controlled by the function $\gamma(M)$) have less influence on the cell cluster sizes when the chemosensitivity parameter $A=B$ is increased (see Figure \ref{fig: wavenumber 2D}). 
}

\begin{figure}[tbhp]
\centering
\subfloat[]{\label{fig: difference solutions 100}\includegraphics[scale=0.28]{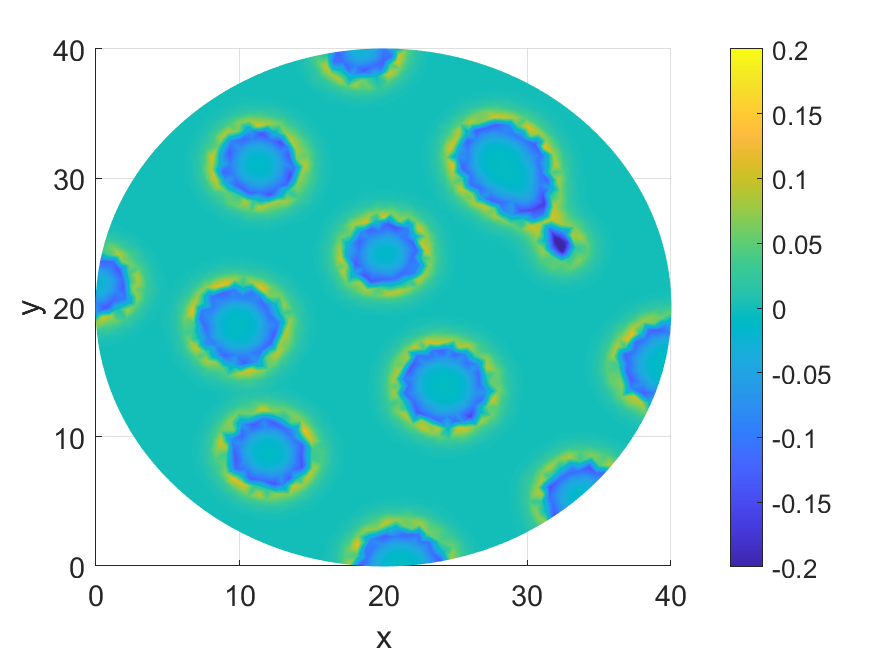}}
    \subfloat[]{\label{fig: difference solutions 1000}\includegraphics[scale=0.28]{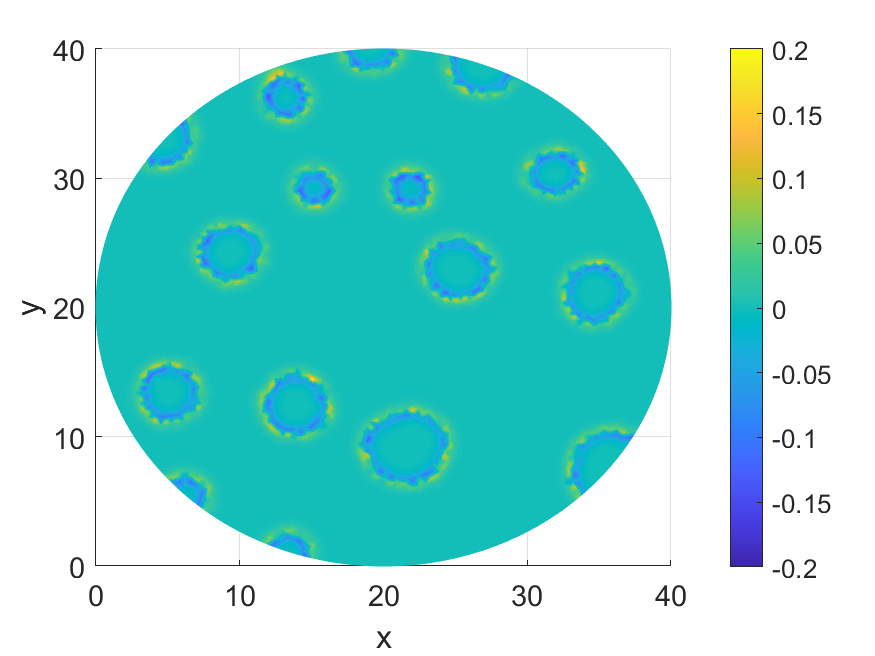}}
    \caption{Difference between the solutions{ obtained without TMZ and with TMZ,  $u(\mathbf{x},T_1)-w(\mathbf{x},T_2)$ for $T_1 = 25$ days, $T_2 = 3.75$ days and (a) $A=B=20$ and (b) $A=B=70$. Here $\gamma(M)=11$, $r_0=0.4 \; \rm{day}^{-1}$ and $u_\textnormal{max} = u_0=0.1$. } }\label{fig: difference solutions}
\end{figure}
\begin{figure}[tbhp]
\centering
\subfloat[]{\label{fig: difference solutions 10 Gauss}\includegraphics[scale=0.28]{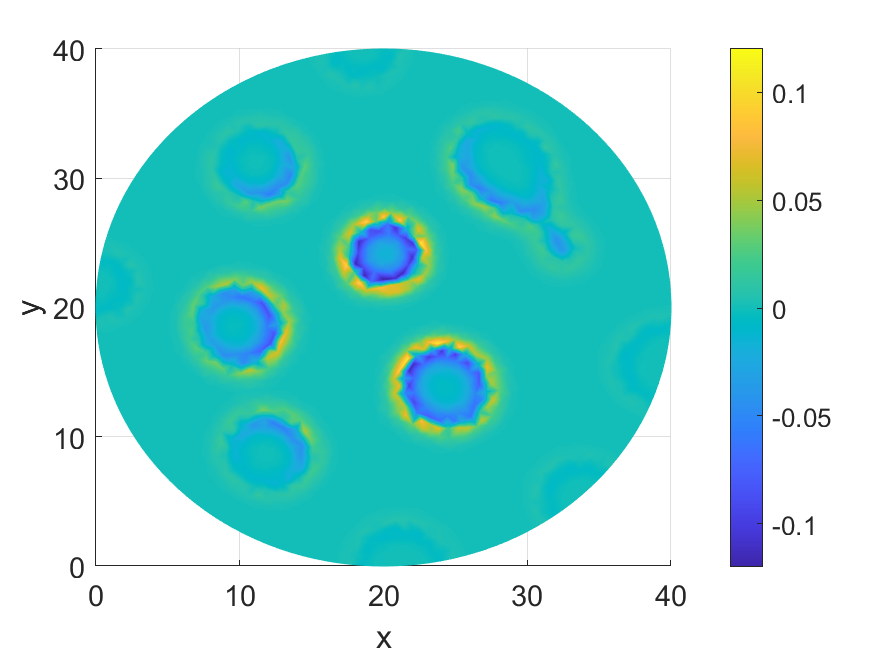}}
    \subfloat[]{\label{fig: difference solutions 30}\includegraphics[scale=0.28]{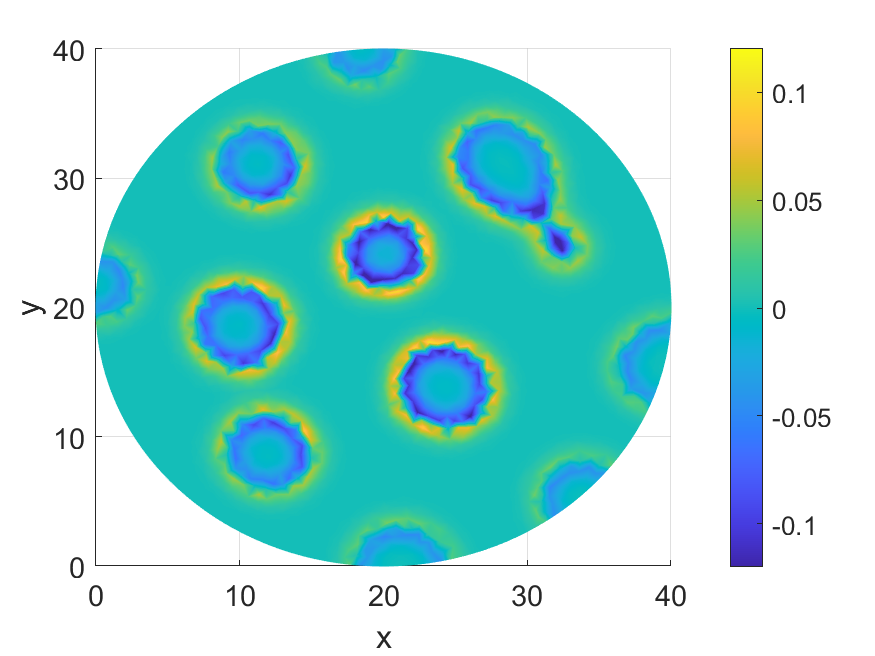}}
        \subfloat[]{\label{fig: difference solutions 100 Gauss}\includegraphics[scale=0.28]{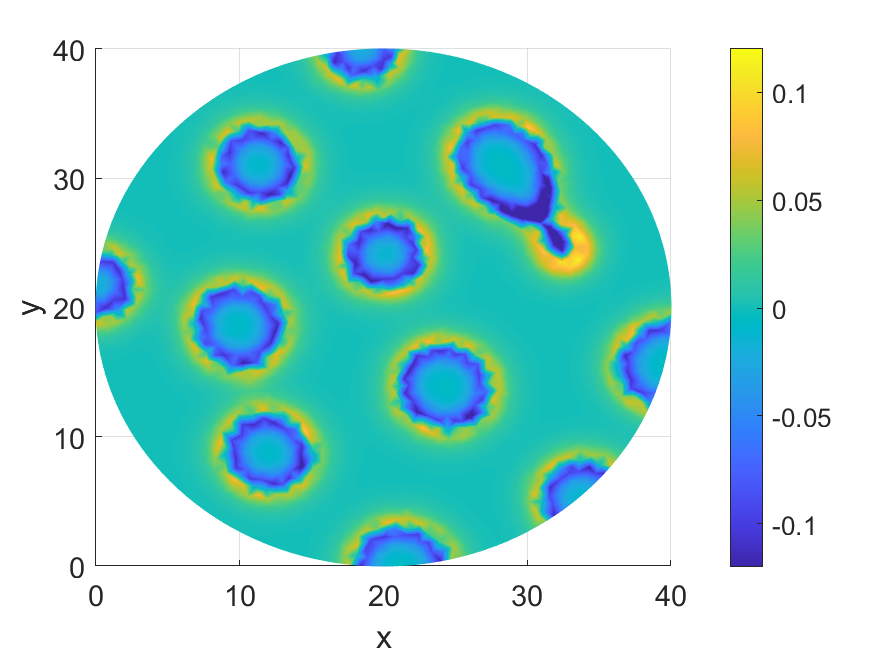}}
    \caption{Difference between the solutions when the initial concentration of the treatment is a Gaussian function centered in the domain. (a) $T_2=1.25$ days,  (b) $T_2=3.75$ days (c) $T_2=12.5$ days for $T_1=25$ days, $r_0=0.4 \; \rm{day}^{-1}$ and $u_0=0.1$.}\label{fig: comparison 2d gauss}
\end{figure}

{We now study the case when the treatment is introduced in the middle of the domain and diffuses in the environment. To this aim, the initial concentration of the treatment is assumed to be a Gaussian function with width $5$ centered in the middle of the domain.  We consider that the treatment is not consumed by the cells in the time scales we are interested in, and choose $\delta_0=0$. In Figure \ref{fig: comparison 2d gauss} we show the evolution of the difference between the two solutions $u(\mathbf{x},T_1)-w(\mathbf{x},T_2)$, where $T_1=25$ days is the time at which the treatment is introduced and $T_2$ is the duration of the treatment. We explore different times $T_2=1.25,\ 3.75,\ 12.5$ days.} For short times, the effect of the treatment is only noticed by the aggregates at the center of the domain and therefore the difference between the two solutions close to the boundaries is zero. As time increases, the concentration of the treatment reaches the whole domain as is observed in Figure \ref{fig: difference solutions 100 Gauss}. 

{To understand if the position of the Gaussian plays a role in our study, we now consider the initial condition for the TMZ to be a Gaussian centered at a boundary of the domain.}
{Figure \ref{fig:evolution-TMZ-border} depicts the difference ${u(x,T_1) - w(x,T_2)}$ when the treatment is introduced at time $T_1 = 25$ days and for three different times $T_2$. We observe that the TMZ concentration diffuses from the left border of the domain, and progressively the aggregates shrink in response to TMZ. }

\begin{figure}[tbhp]
  \centering
  \subfloat[]{\label{fig:appendix-evolution-beg}\includegraphics[width=0.32\linewidth]{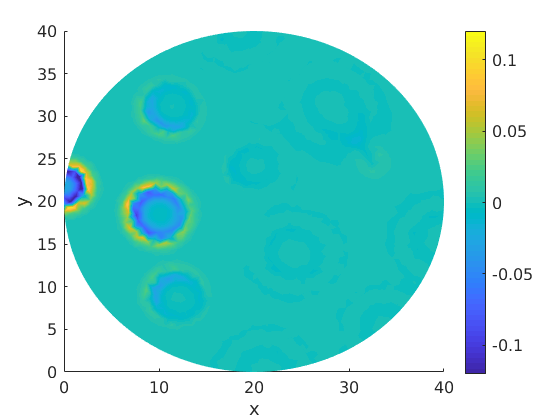}}
  \subfloat[]{\label{fig:appendix-evolution-1sec}\includegraphics[width=0.32\linewidth]{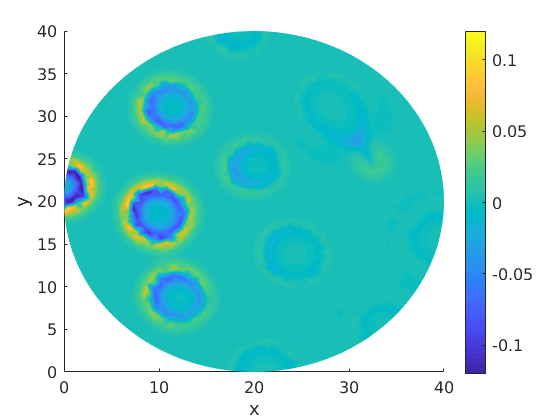}}
  \subfloat[]{\label{fig:appendix-evolution-end}\includegraphics[width=0.32\linewidth]{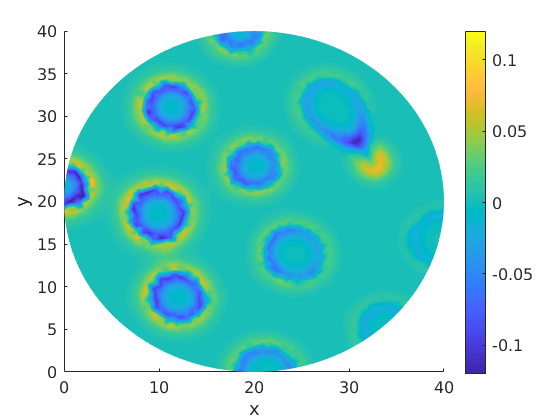}}
  \caption{{Evolution of the difference ${u(x,T_1) - w(x,T_2)}$ for $T_1 = 25$ days, at three different times: (a) $T_2=1.25$ days, (b) $T_2=3.75$ days, (c) $T_2=12.5$ days.} }
  \label{fig:evolution-TMZ-border}
\end{figure}

Finally, we also study the effect of the treatment at earlier stages of the formation of the aggregates. Figure \ref{fig: TMZ different times} shows the different patterns obtained during the formation of the aggregates at different times (part \textbf{P1}), top row, and the corresponding effect of the treatment (part \textbf{P2}), bottom row. For example, introducing the treatment at $T_1=6$ days leads to a significant reduction of the size of the pattern with a reasonably low concentration of the treatment. Identifying this specific time in real patients could make the treatment much more effective and reduce the spread of the cancer cells.

\begin{figure}[tbhp]
\centering
\subfloat[]{\label{fig: noTMZ 30 }\includegraphics[scale=0.28]{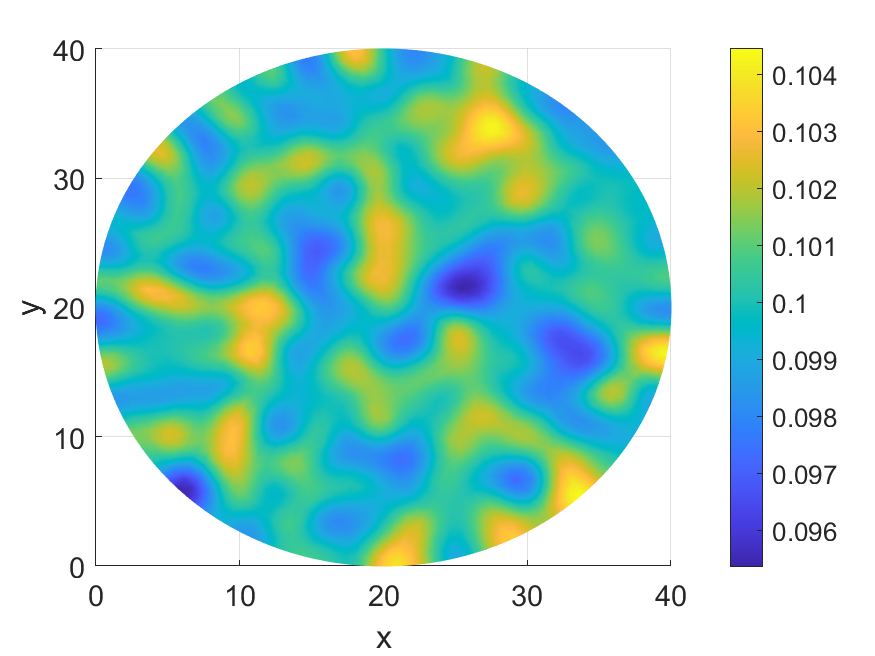}}
    \subfloat[]{\label{fig: noTMZ 500 }\includegraphics[scale=0.28]{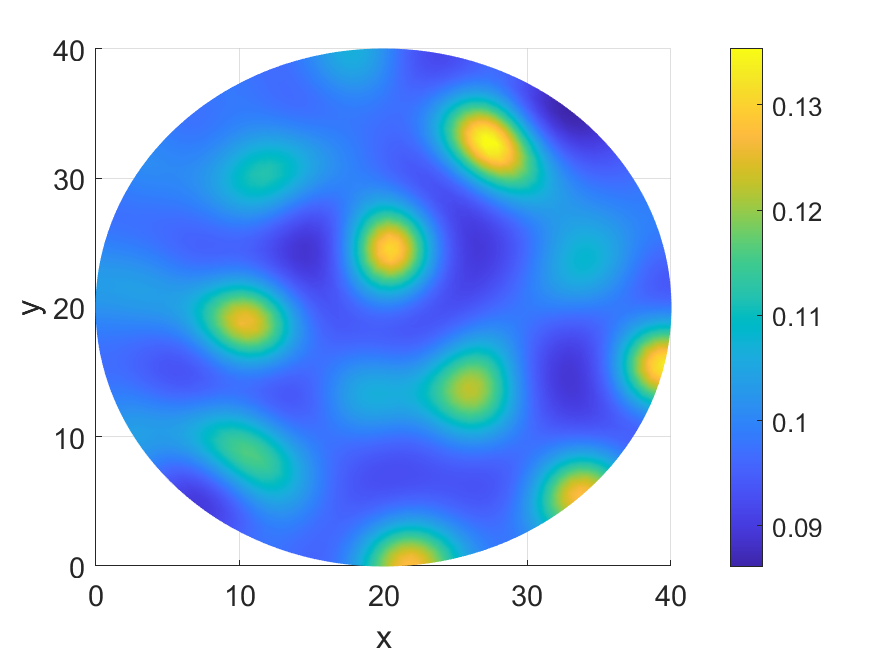}}
        \subfloat[]{\label{fig: noTMZ 1000 }\includegraphics[scale=0.28]{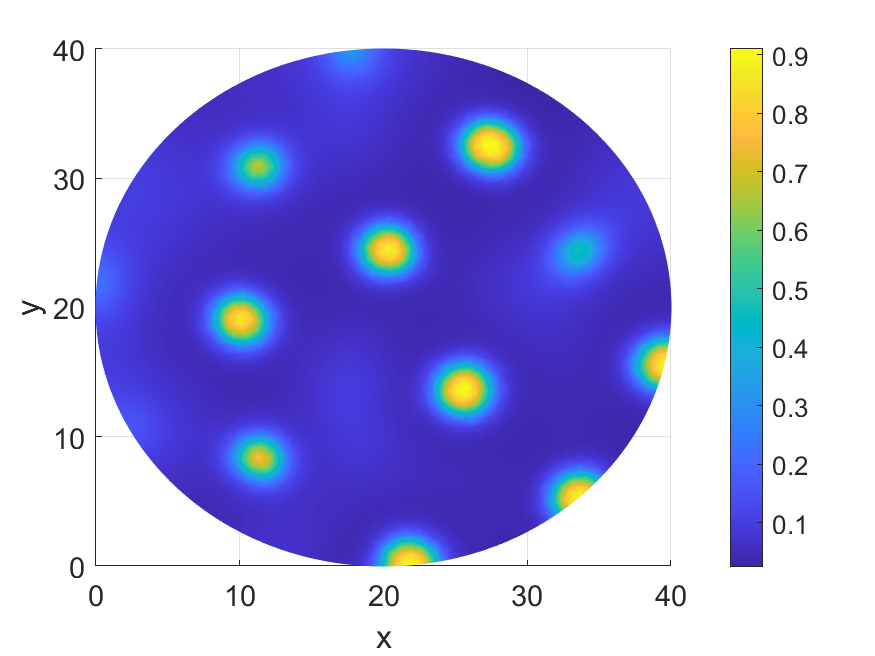}}
        
        \subfloat[]{\label{fig: TMZ 30 }\includegraphics[scale=0.28]{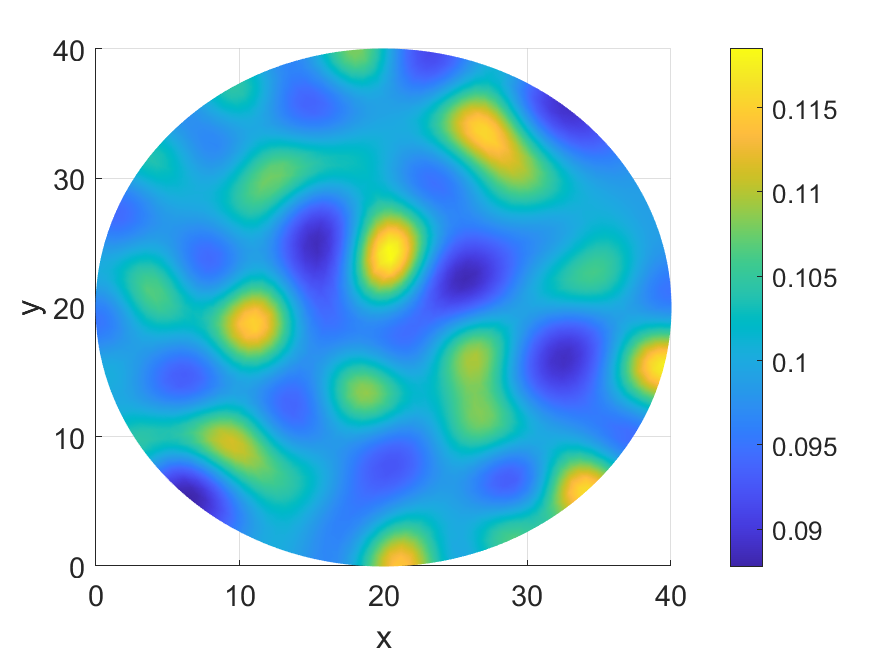}}
    \subfloat[]{\label{fig: TMZ 500 }\includegraphics[scale=0.28]{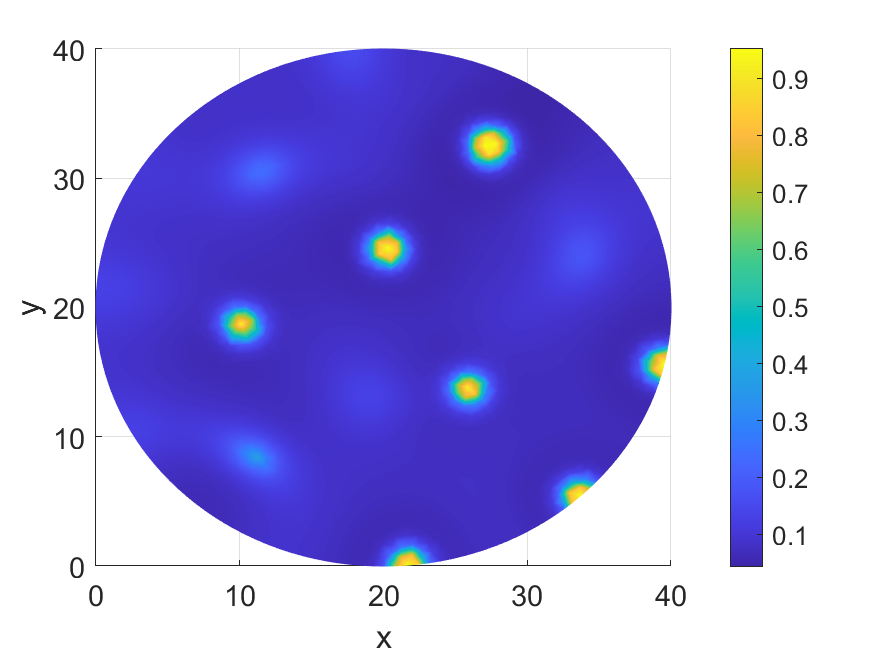}}
        \subfloat[]{\label{fig: TMZ 1000 }\includegraphics[scale=0.28]{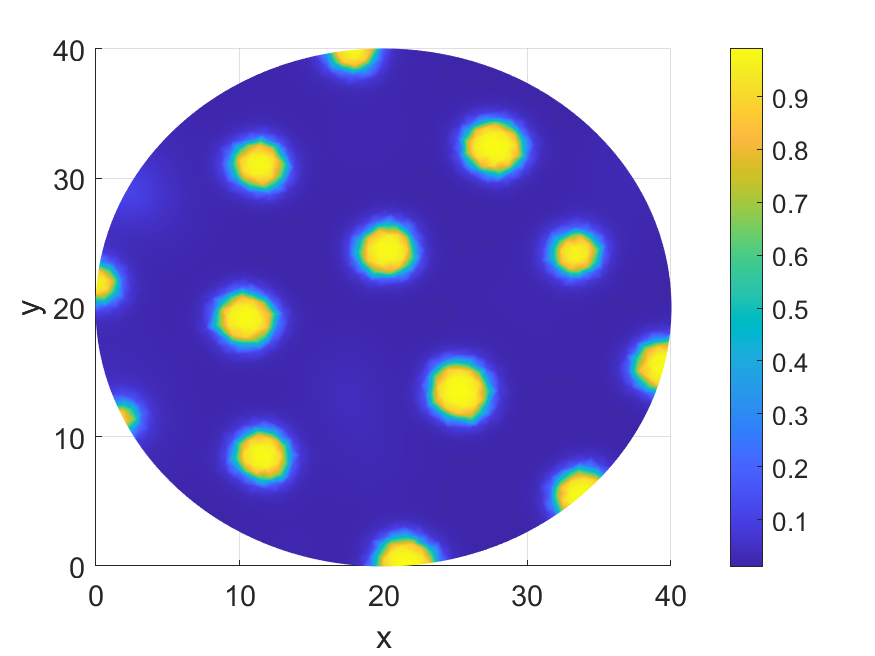}}
        \caption{Effect of the treatment at different times while the aggregates are forming for $A=20$. Top row: solution obtained for \textbf{P1} at (a) $T_1=3.75$ days, (b) $T_1=6.2$ days and (c) $T_1=12.5$ days. Bottom row: solution  at time $T = T_1 + 25$ days, when the treatment has been introduced at times (d) $T_1=3.75$ days, (e) $T_1=6.2$ days and (f) $T_1=12.5$ days.}\label{fig: TMZ different times}
\end{figure}

\section{Discussion of results and perspectives}\label{sec: discussion and perspectives}

{In this paper we propose a mechanism for the effect of certain treatments on tumours formed by a chemotaxis type self-organisation phenomenon. Inspired by the experiments concerning the action of TMZ on Glioblastoma cells mentioned in Section \ref{sec: experiments}, we considered the particular case of a non-cytotoxic treatment which could induce changes in the mechanical properties of individual tumour cells by making them pass from rigid bodies to {semi-elastic} entities. We explored two scenarios: in the first one, only cell's plasticity is impacted by the treatment, and in the second one, the treatment has a double effect of preventing cell proliferation as well as changing cell mechanics.} 

{Under these hypotheses, we obtained a modified version of the Keller-Segel model, known as the nonlinear volume-filling approach for cell motion, where the cell mechanical properties are taken into account in the form of the so-called squeezing probability. In the nonlinear volume-filling Keller-Segel model, this squeezing probability function could be related to the amplitude of the transport term towards zones of high chemoattractant concentration (chemosensitivity), as well as with the (nonlinear) diffusion coefficient. }

{By performing a linear stability analysis, we are able to detect and characterise the parameter ranges for which the homogeneous distribution is unstable, \emph{i.e.} the ranges for which patterns appear as the result of the dynamics. } {We show that the emergence of patterns without treatment (\emph{i.e.} when cells act as rigid bodies) is driven by a fine interplay between the chemotactic sensitivity, which tends to aggregate the cells, and the diffusion, together with the logistic growth, which tend to smoothen the solution. We are able to compute the critical chemosensitivity value above which the system self-organizes into aggregates, and characterise the size of the aggregates as a function of the model parameters.} 

{Under treatment, \emph{i.e.} when cells behave as {semi-elastic} entities, we show that the critical value of the chemosensitivity above which patterns emerge is smaller than that in control cultures, showing that as cells become more elastic, they tend to aggregate more easily than when they behave as rigid entities. }

{We are able to completely characterise the size of the patterns and show that {semi-elastic} cells create smaller aggregates than rigid entities for the same value of the chemosensitivity. These results suggest that the mechanical properties of individual cells have a huge impact on the shape and size of the aggregated patterns at the population level.} 

{Moreover, we show that the ratio between the tight packing cell density and the carrying capacity of the TME plays a major role in the size and shape of the obtained patterns. For large values of this ratio, the aggregation abilities of the system are essentially driven by the chemotactic transport term while the individual cell mechanical behaviour has little impact on the shape and size of the patterns. However, for smaller values of this ratio, \emph{i.e.} when the tight packing density is closer to the carrying capacity of the environment, cell mechanics has a huge influence on the behaviour of the population. }

{These results are confirmed by numerical experiments in 1D and 2D for which, given an initial perturbation of the homogeneous cell distribution, we observe the emergence of cell aggregates and we recover the critical values of the chemosensitivity predicted by the stability analysis. We obtain a very good correspondence between the simulations and the theoretical predictions, for the appearance of patterns as well as their size. }

{By performing numerical simulations of the systems \eqref{eq: final system 1} and  \eqref{eq: nondimensional second phase}, we recover the experimental observations: introducing the treatment  on already-formed aggregates leads to a rapid formation of more compact patterns. As the treatment diffuses in the domain (changing locally the cell mechanical properties as it goes), it sharpens the border of the cell aggregates and leads to denser and well-separated clusters.} 

{While the border sharpening of the clusters is independent on whether proliferation is activated or not during treatment, the shrinking of the clusters is clearer when the treatment has the double effect of changing the cell mechanical properties as well as blocking cell proliferation. Indeed, when proliferation is still active in the presence of the treatment, we observe the merging of existing clusters and this results in aggregates being larger than before treatment. 
These results suggest a possible mechanism for the shrinking of the aggregates observed under the experimental conditions described in Section \ref{sec: experiments}: TMZ might not only stop cell proliferation, but might also generate a stress in the environment to which cells respond to by changing their mechanical state. }

{While alterations of mechanical properties, around or inside the tumour, are common in solid tumours including GBM, {the question of the nature and the regulation of cancer cells through mechano-sensitive pathways are largely unknown. Recently, in a Drosophila model, glioma progression has been associated to a regulatory loop mediated by the mechano-sensitive ion channel Piezo1 and tissue stiffness \cite{chen2018feedforward}.}
{A direct perspective of these works consists in verifying the potential mechano-sensing effect of TMZ proposed in the present model, through direct measures in real systems by studying the mechanical properties of individual GBM cells which have been exposed to TMZ treatment. }
{The targeting of mechano-sensitive pathways after TMZ treatment may provide new therapeutic angles in GBM and in more general settings.}}

{Moreover, it would be interesting to identify other clinical settings where the effects of the treatment are similar to those of TMZ in GBM, and to check if the effects are due to changes in the tumor cell properties corresponding to the general hypothesis of the model constructed in this work.} 
{In the future, exhaustive quantitative comparison with experiments will allow for systematic choice of parameters and validation of the mechanisms we propose here. From a biological point of view, a natural sequel of this work consists in studying the coupled effect of TMZ and irradiation. Indeed, even if TMZ alone seems not to suffice to decrease the tumour mass, the coupling of TMZ treatment with irradiation has been shown to have more efficient effects than irradiation alone \cite{stupp2009effects,barazzuol2012vitro,van2000survival}. It would also be interesting to introduce a second treatment with cytotoxic effects in this model, to study whether the mechanical changes of individual cells induced by TMZ could explain the better response of the system to irradiation treatments.}
{Finally on a modelling viewpoint, it would be interesting to extend the model to take into account cell-cell adhesion in the spirit of \cite{Painter2006}, in order to study the effects of cell-cell adhesion \emph{vs.} individual cell mechanics on the aggregation properties of cell populations. }

\newpage
\appendix

\section{Stability analysis}
\label{app: stability analysis}

\noindent\textbf{First part of the experiments} {We first observe that the homogeneous distributions $u(\mathbf{x},t) = u^*$ and $c(\mathbf{x},t) = c^*$ are steady-states solutions of system \eqref{eq: final system 1} for $u^*$ and $c^*$ such that $f(u^*) = 0$ and $g(u^*,c^*) = 0$. In order to study their stability, we consider the system without spatial variations}
\begin{equation}
    \partial_tu=f(u)\ ,\ \ \zeta\partial_tc=g(u,c)\ ,
\end{equation}
and linearize the solution at $(u^*,c^*)$. We obtain
\begin{equation}
\partial_t \sigma = G\sigma\ ,\ \ \textnormal{where}\ \ \sigma= \begin{pmatrix}u-u^*\\c-c^*
    \end{pmatrix}\ \ \textnormal{and}\ \ 
G=\begin{pmatrix}f_u^*& 0 \\ \zeta^{-1} g_u^* & \zeta^{-1} g_c^*
    \end{pmatrix}\ ,
\end{equation}
{where the quantities $f_u^*,\ g_u^*$ and $g_c^*$ are the linearization slopes of $f$ and $g$: $f_u^* = f'(u^*),\ g^*_u = \partial_u g(u^*,c^*),\ g^*_c = \partial_c g(u^*,c^*)$.} The steady state is linearly stable if  $\textnormal{tr}(G)<0$ and $\textnormal{det}(G)>0$, which imposes the following constraints on the kinetic functions $f(u)$ and $g(u,c)$,
\begin{equation}
    f_u^*+\zeta^{-1}g_c^*<0\ \ \textnormal{and}\ \ f_u^*g_c^*>0\ .\label{eq: instability condition 1}
\end{equation}

{Note that in our case, $f^*_u = -r_0,\ g^*_u = 1,\ g^*_c = -1$ so the conditions are satisfied. }

We now {go back to} the full chemotactic system (\ref{eq: syst linear stability1}).
In order to investigate {the stability of the homogeneous steady-state, \emph{i.e.} the ability of the system to create patterns, we introduce a small parameter $\varepsilon \ll 1$ and write}
\begin{equation}
u=u^*+\varepsilon\tilde{u}(\mathbf{x},t)\ ,\ \ \ c=c^*+\varepsilon\tilde{c}(\mathbf{x},t)\ .\label{eq: expansion}
\end{equation}
We substitute (\ref{eq: expansion}) into (\ref{eq: syst linear stability1}) and, computing the first order terms with respect to $\varepsilon$ and neglecting higher order terms, the linearized system reads
\begin{equation}
\begin{aligned}
    \partial_t u & =\Delta u-A\phi_1(u^*)\Delta c+uf^*_u+cf^*_c\ ,\\
    \zeta\partial_t c & = \Delta c+ug^*_u+cg^*_c\ ,\label{eq: linearized system}
\end{aligned}
\end{equation}
where  $\phi_1(u^*)=u^*q_1(u^*)$. 
{We now look for perturbations of the form}
\begin{equation}
    u(\mathbf{x},t)=\sum_k a_k(t)\psi_k(\mathbf{x})\ \ \textnormal{and} \ \ \ c(\mathbf{x},t)=\sum_k b_k(t)\psi_k(\mathbf{x})\ ,\label{eq: fourier modes}
\end{equation}
where $(\psi_k)_{k\geq 1}$ is an orthonormal basis of $L^2(\rm{\Omega})$ and satisfies the following spatial eigenvalue problem
\begin{equation}
    -\Delta \psi_k=k^2\psi_k\ , \ \ \frac{\partial\psi_k}{\partial \eta}=0\ .\label{eq: psi}
\end{equation}
Then, the linearized system (\ref{eq: linearized system}) can be written as
\begin{equation}
    \begin{aligned}
     \partial_t (a_k)&=-a_kk^2+A\phi_1(u^*)b_kk^2+a_kf^*_u+b_kf^*_c\ ,\\
\zeta   \partial_t( b_k) & = - b_kk^2+a_kg^*_u+b_kg^*_c\ ,
    \end{aligned}\label{eq: system 42}
\end{equation}
where $k$ is the spatial eigenfunction, also called the wavenumber and $1/k$ is proportional to the wavelength $\omega$. In matrix form we can write (\ref{eq: system 42}) as $\partial_t X_k(t)=P_k(t)X_k(t)$ where
\begin{equation}
    X_k=\begin{pmatrix}
           a_k \\
           b_k 
         \end{pmatrix}\ ,\ \ \ P_k=\begin{pmatrix}
-k^2+f^*_u & A\phi_1(u^*)k^2+f^*_c \\
\zeta^{-1}g^*_u & \zeta^{-1}(-k^2+g^*_c)
\end{pmatrix}\label{eq: matrix form}\ .
\end{equation}
{
\begin{remark}
Since the solutions of the eigenvalue problem (\ref{eq: psi}) are simply sines and cosines, the ``size'' of various spatial patterns is measured by the wavelength of the trigonometric functions. For example, in one dimension when $0<x<L$, $\psi\propto \cos(n\pi x/L)$ and the wavelength is $\omega=1/k=L/n\pi$, where $n\in\mathbb{Z}$.
\end{remark}
}
If the matrix $P_k$ has eigenvalues with positive real part, then the homogeneous steady state $(u^*,\ c^*)$ is unstable, resulting in pattern formation. The characteristic polynomial related to (\ref{eq: matrix form}) is given by $\ell^2+a(k^2)\ell+b(k^2)=0$ where 
\begin{align}
    a(k^2) & = (1+\zeta^{-1})k^2-(f^*_u+\zeta^{-1}g^*_c)\ ,\label{eq: a expression}\\
    b(k^2) & = \zeta^{-1}k^4-\zeta^{-1}(g^*_c+f^*_u+g^*_uA\phi_1(u^*))k^2+\zeta^{-1}f_u^*g_c^*\ .\label{eq: b expression}
\end{align}
The eigenvalues $\ell$ determines the temporal growth of the eigenmodes, {and we require $\mathcal{R}e(\ell(k^2))>0$ for the homogeneous steady state to be unstable. Note that we only look for the eigenmodes $k\neq 0$ since we already guaranteed that the steady state is  stable in the absence of spatial perturbations, \emph{i.e.} $\mathcal{R}e(\ell(k^2=0))<0$ in (\ref{eq: instability condition 1})}. 

From the conditions (\ref{eq: instability condition 1}), we know that $a(k^2)>0$, hence the instability can  only occur if $b(k^2)<0$ for some $k$ so that the characteristic polynomial has one positive and one negative root. 
This implies
\begin{equation}
k^4-(g^*_c+f^*_u+g^*_uA\phi_1(u^*))k^2+f_u^*g_c^*<0\ .\label{eq: condition}
\end{equation}
We also know from (\ref{eq: instability condition 1}) that $f_u^*g_c^*>0$, then a necessary (but not sufficient) condition for $b(k^2)<0$ is 
\[
g^*_c+f^*_u+g^*_uA\phi_1(u^*)>0\ .
\]
{
\begin{remark}
The bifurcation between spatially stable and unstable modes occurs when the critical expression $b_\textnormal{min}(k^2_\textnormal{min})=0$ is satisfied.
\end{remark}
}

Moreover, to satisfy (\ref{eq: condition}) the minimum $b_\textnormal{min}$ must be negative \cite{murray2001mathematical}. Differentiation with respect to $k^2$ in (\ref{eq: b expression}) leads to
\begin{equation}
    b_\textnormal{min}(k^2_\textnormal{min})=-\frac{(g_c^*+f_u^*+g_u^*A\phi(u^*))^2}{4}+f_u^*g_c^*\ .\label{eq: critical value b}
\end{equation}
Hence,  the condition $b_\textnormal{min}<0$ implies
\begin{equation}
    g_c^*+f_u^*+g_u^*A\phi_1(u^*)>2\sqrt{(f_u^*g_c^*)}\ .
\end{equation}

To summarise, we have obtained the following conditions for the generation of spatial patterns for the chemotaxis system (\ref{eq: syst linear stability1}),
\begin{equation}
    \begin{aligned}
    f_u^*+\zeta^{-1}g_c^*&<0\ ,\  f_u^*g_c^*>0\ ,\  \zeta^{-1}g^*_c+f^*_u-\zeta^{-1}g^*_uA\phi_1(u^*)>0\ , \\ &g_c^*+f_u^*+g_u^*A\phi_1(u^*)>2\sqrt{f_u^*g_c^*}\ .
    \end{aligned}\label{eq: stability properties}
\end{equation}

From the analysis in this section, and using the particular forms of $\phi_1(u)$, $f$ and $g$ as in (\ref{eq: syst linear stability1}), it is easy to see that the spatially homogeneous steady states are $(0,0)$ and $(u_{\textnormal{max}}, u_{\textnormal{max}})$. We can check that $(0,0)$ is an unstable steady state, therefore we only work with $(u_{\textnormal{max}}, u_{\textnormal{max}})$ which, on the contrary, is stable.  The first and second properties in (\ref{eq: stability properties}) are immediately satisfied, \emph{i.e.}, $-(r_0+\zeta^{-1})<0$ and $r_0/\zeta>0$, respectively.
Finally, we have to check that the third and fourth conditions are satisfied as well. We have that
\begin{equation}
    - 1-r_0+A\phi_1(u^*)>2\sqrt{r_0}\ .\label{eq: necesary condition}
\end{equation}
Therefore, (\ref{eq: necesary condition}) is a necessary condition for pattern formation for the original system (\ref{eq: chemotaxis detail}). Considering $A$ as a bifurcation parameter, we can obtain a critical value $A^c$, so that we observe pattern formation if $A>A^c$. From (\ref{eq: necesary condition}) we get
\begin{equation}
    A^c=\frac{2\sqrt{r_0}+ 1+r_0}{ u_{\textnormal{max}}\left(1-\frac{u_{\textnormal{max}}}{\bar{u}} \right)}\ .\label{eq: critical chemosensitivity}
\end{equation}
The corresponding critical wavenumber $k^2_c$ is obtained from (\ref{eq: critical value b}) using (\ref{eq: necesary condition}) as follows,
\begin{equation}
    k^2_c=\frac{g_c^*+f_u^*+g_u^*A^c\phi_1(u^*)}{2}=\sqrt{f_u^*g_c^*}=\sqrt{r_0}\ .\label{eq: k critic 1}
\end{equation}
This means that, within the unstable range, $\mathcal{R}e(\ell(k^2))>0$ has a maximum wavenumber given by $k_c^2$. The range of linearly unstable modes $k_1^2<k^2<k_2^2$ is obtained from $b(k^2)=0$,
\begin{align}
    k_1^2=&\frac{-m-\sqrt{m^2-4 f_u^*g_c^*}}{2}<k^2<k^2_2=\frac{-m+\sqrt{m^2-4 f_u^*g_c^* }}{2}\ ,\label{eq: unstable modes}
\end{align}
where $m=-(g_c^*+f_u^*+g_u^*A\phi_1(u^*))$.\\

\noindent\textbf{Second part of the experiments} Following the same steps as  before  we linearise the system (\ref{eq: nondimensional second phase}) to get
\begin{equation}
    \begin{aligned}
    \partial_t w & = {{D}}_2(w^*,M^*)\Delta w-B\phi_2(w^*,M^*)\Delta c+wf^*_w\ ,\\ \zeta\partial_t c & = \Delta c+wg^*_w+cg^*_c\ , \\
    m\partial_t M & = \Delta M-\delta_0 w\ ,
    \end{aligned}
\end{equation}
where
$$D_2(w^*,M^*)=1+(\gamma(M)^*-1)\Bigl( \frac{w^*}{\bar{u}}\Bigr)^{\gamma(M)^*}\ ,\ \ \phi_2(w^*,M^*)=w\Bigl(1-\Bigl(\frac{w^*}{\bar{u}} \Bigr)^{\gamma(M)^*}\Bigr)\ ,$$
and $\gamma(M)^*$ is given by (\ref{eq: gamma M}) evaluated at $M^*$.
As in (\ref{eq: fourier modes}), we let
\begin{equation}
    w(\mathbf{x},t)=\sum_ka_k(t)\psi_k(\mathbf{x})\ ,\ \ c(\mathbf{x},t)=\sum_kb_k(t)\psi_k(\mathbf{x})\ ,\ \ M(\mathbf{x},t)=\sum_kc_k(t)\psi_k(\mathbf{x})\ ,
\end{equation}
where $\psi_k(\mathbf{x})$ satisfies (\ref{eq: psi}) and we obtain a system $\partial_t X_k(t)=P_k(t)X_k(t)$ where 
\begin{equation}
    X_k=\begin{pmatrix}
    a_k\\ b_k\\ c_k
    \end{pmatrix}\ ,\ \ P_k=\begin{pmatrix} -{{D}_2}(w^*,M^*)k^2+f_w^* & {B}{\phi}_2(w^*,M^*)k^2 & 0\\ \frac{1}{\zeta} g^*_w & \frac{1}{\zeta}(-k^2+g^*_c) & 0\\ -\frac{\delta_0}{m} & 0 & -\frac{1}{m}k^2 
    \end{pmatrix}\ .\label{eq: stability of the treatment}
\end{equation}
Similar to the previous section, the characteristic polynomial is given by $a(k^2)\ell^3+b(k^2)\ell^2+c(k^2)\ell+d(k^2)=0$ where $a(k^2)=-1$ and 
\begin{align}
    b(k^2) & = -\Bigl({D_2}+\frac{1}{\zeta}+\frac{1}{m}\Bigr)k^2+f_w^*+\frac{g_c^*}{\zeta} \ ,\\ c(k^2) & = -\Bigl(\frac{{D_2}}{\zeta}+\frac{{D_2}}{m}+\frac{1}{m\zeta}\Bigr)k^4+\Bigl(\frac{g_c^*{D_2}}{\zeta}+\frac{g_c^*}{\zeta m}+\frac{f_w^*}{\zeta}+\frac{f^*_w}{m}+\frac{g_w^*B\phi_2}{\zeta}\Bigr)k^2\nonumber\\ & \ \ \ -\frac{f^*_wg^*_c}{\zeta} \ ,\\ d(k^2) & =-\frac{{D_2}}{ m\zeta} k^6+\frac{1}{ m\zeta}( {D_2}g_c^*+g^*_wB\phi_2+f^*_w)k^4+\frac{1}{m\zeta }(-f^*_wg^*_c)k^2\ .
\end{align}

In general, the stability analysis for this cubic polynomial will require the Ruth–Hurwitz stability criterion \cite{hurwitz1964conditions} which states that the steady state is unstable if the coefficients of $a(k^2)\ell^3+b(k^2)\ell^2+c(k^2)\ell+d(k^2)=0$ satisfy the condition
\[
\frac{1}{(a(k^2))^2}\left(b(k^2)c(k^2)-a(k^2)d(k^2)\right)<0\ .
\]
However, from (\ref{eq: stability of the treatment}) we observe that one of the eigenvalues of the matrix $P_k$ is given by $\ell_1=\frac{-k^2}{m}<0$. The remaining two eigenvalues can be computed from the upper-left matrix
\begin{equation}
    \begin{pmatrix} -{D_2}(w^*,M^*)k^2+f^*_w & {B}\phi_2(w^*,M^*)k^2 \\ \frac{1}{\zeta} g^*_w & \frac{1}{\zeta}(-k^2+g^*_c)  \label{eq: other matrix}
    \end{pmatrix}\ ,
\end{equation}
following the same analysis as for the case without TMZ.

The characteristic polynomial $\ell^2+\bar{a}(k^2)\ell+\bar{b}(k^2)=0$ related to (\ref{eq: other matrix}) has coefficients
\begin{align}
    \bar{a}(k^2)&=\Bigl({{D_2}(w^*,M^*)}+\frac{1}{\zeta} \Bigr)k^2-f^*_w-\frac{g_c^*}{\zeta}\ ,\label{eq: a bar}\\ \bar{b}(k^2)&=\frac{{D_2}(w^*,M^*)}{\zeta}k^4-\Bigl(\frac{{D_2}(w^*,M^*)g_c^*}{\zeta}+\frac{B\phi_2(w^*,M^*)g_n^*}{\zeta}+\frac{f_w^*}{\zeta}\Bigr)k^2\nonumber\\  & \ \ \ +\frac{f^*_wg^*_c}{\zeta}\label{eq: b bar}\ .
\end{align}
For the steady state to be unstable we require, as before, that $\mathcal{R}e(\ell(k^2))>0$. Since $\bar{a}(k^2)>0$ the instability can only occur if $\bar{b}(k^2)<0$. Computing $\frac{\diff \bar{b}(k^2)}{\diff k^2}=0$  from (\ref{eq: b bar}) we obtain
\begin{equation}
    k^2_{\textnormal{min}}=\frac{D_2(w^*,M^*)g^*_c+g^*_wB\phi_2(w^*,M^*) +f^*_w}{2D_2(w^*,M^*)}\ .
\end{equation}
Hence from the condition $\bar{b}_{\textnormal{min}}(k^2_{\textnormal{min}})<0$ we get
\begin{equation}
    D_2(w^*,M^*)g^*_c+g^*_wB\phi_2(w^*,M^*)+f^*_w>\sqrt{4D_2(w^*,M^*)f^*_wg^*_c}\ .\label{eq: pattern formation condition}
\end{equation}
The spatially homogeneous steady state is ${(w^*,\ c^*,\ M^*)=(u_\textnormal{max},\ u_\textnormal{max},\ M_s)}$, where  \newline$ M_s=|\rm{\Omega}|^{-1}\int_{\rm{\Omega}} M(\mathbf{x},0)\diff\mathbf{x}$. Therefore, from (\ref{eq: pattern formation condition}) we obtain a critical constant $B^c$ so that for any $B>B^c$ we observe pattern formation. This critical constant is given by
\begin{equation}
    B^c=\frac{2\sqrt{\bar{r}_0D_2(u_\textnormal{max},M_s)}+D_2(u_\textnormal{max},M_s)+\tilde{r}_0}{u_\textnormal{max}\Bigl(1-\Bigl(\frac{u_\textnormal{max}}{\bar{u}} \Bigr)^{\gamma_{M_s}} \Bigr)}\ ,\label{eq: B critic}
\end{equation}
 where 
 \begin{equation}
 D_2(u_\textnormal{max},M_s)=1+(\gamma_{M_s}-1)\Bigl(\frac{u_\textnormal{max}}{\bar{u}} \Bigr)^{\gamma_{M_s}}\ .\label{eq: diffusion coefficient}
 \end{equation}
The corresponding critical wavemode is given by
\begin{equation}
    k^2_c=\frac{D_2(u_\textnormal{max},M_s)g_c^*+g_w^*B^c\phi_2(u_\textnormal{max},M_s)+f^*_w}{2D_2(u_\textnormal{max},M_s)}=\frac{\sqrt{D_2(u_\textnormal{max},M_s)(f^*_wg^*_c)}}{D_2(u_\textnormal{max},M_s)}\ .\label{eq: k critic 2}
\end{equation}
Finally, the unstable modes are $k^2<k^2_c$, where from $\bar{b}(k^2)=0$ we get
\begin{align}
    k_1^2=&\frac{-\bar{m}-\sqrt{\bar{m}^2-4D_2(u_\textnormal{max},M_s)(f_w^*g_c^*)}}{2D_2(u_\textnormal{max},M_s)}<k^2<k^2_2\nonumber\\ &=\frac{-\bar{m}+\sqrt{\bar{m}^2-4D_2(u_\textnormal{max},M_s)(f_w^*g_c^*)}}{2D_2(u_\textnormal{max},M_s)}\ ,
\end{align}
for $\bar{m}=-(D_2(u_\textnormal{max},M_s)g_c^*+g_w^*B\phi_2(u_\textnormal{max},M_s)+f^*_w)$.

\section{Description of the numerics}\label{app: numerics description}
We denote $H^1(\rm{\Omega})=W^{1,2}(\rm{\Omega})$  the usual Sobolev space and the standard $L^2$ inner product is denoted by $(\cdot,\cdot)$.
Let $\mathcal{T}^h$, {$h>0$}, be a quasi-uniform mesh of the domain $\rm{\Omega}$ consisting of $N$ disjoint piecewise linear mesh elements $K$ such that the discretized domain $\overline \rm{\Omega}_h = \bigcup_{K\in \mathcal{T}^h} \overline K$.  We denote the total number of nodes by $N_h$.

Let $h_K := \text{diam}(K)$ and $h = \max_{K} h_K$ 
and for $d=2$, we choose linear triangular elements. In addition, we assume that the mesh is acute \emph{i.e.} for $d=2$, each angle of the triangles can not exceed $\frac{\pi}{2}$. We must stress that for $d=2$ since the domain $\rm{\Omega}$ is circular, a small error of approximation is committed using $\Omega_h$.
We consider the standard finite element space associated with $\mathcal{T}^h$
\begin{equation}
    S^h := \{ \chi \in C(\overline{\rm{\Omega}}):  \restr{\chi}K \in \mathbb{P}^1(K), \quad \forall K \in \mathcal{T}^h \} \subset H^1(\rm{\Omega})\ , \label{eq: Sh}
\end{equation}
where $\mathbb{P}^1(K)$ denotes the space of first-order polynomials on $K$. 
Let $N_h$ be the total number of nodes of $\mathcal{T}^h$, $J_h$ the set of nodes and $\{x_j\}_{j=1,\dots,N_h}$ their coordinates. We call $\{ \chi_j\}_{j=1,\dots,N_h}$ the standard Lagrangian basis functions associated with the spatial mesh. 
We denote by $\pi^h:C(\overline{\rm{\Omega}})\to S^h$ the standard Lagrangian interpolation operator. 
We also need the lumped scalar product to define the problem 
\begin{equation*}
    (\eta_1,\eta_2)^h = \int_{\rm{\Omega}} \pi^h\left(\eta_1(x) \eta_2(x)\right)\diff x \equiv \sum_{x_j \in J_h} (1,\chi_j) \eta_1(x_j) \eta_2(x_j)\ , \quad \eta_1, \eta_2 \in C(\overline{\rm{\Omega}})\ .
\end{equation*}
We define the standard mass and stiffness finite element matrices as $G$ and $K$, where
\begin{equation*}
	G_{ij} = \int_{\rm{\Omega}} \chi_i  \chi_j \, \text{d}x, \quad \text{ for } i,j = 1, \dots, N_h\ ,
\end{equation*} 
\begin{equation*}
	K_{ij} = \int_{\rm{\Omega}} \nabla \chi_i \nabla \chi_j \, \text{d}x, \quad \text{ for } i,j = 1, \dots, N_h\ .
\end{equation*} 
In the following finite element approximation of the Keller-Segel problem, the mass matrix is lumped, \emph{i.e.} the matrix becomes diagonal with each term being the row-sum of the corresponding row of the standard mass matrix,
\begin{equation*}
	G_{l,ii} := \sum_{j=1}^{N_h} G_{ij}\ , 	 \quad \text{ for } i = 1, \dots, N_h\ .
\end{equation*}

 
Given $N_T \in \mathbb{N}^\star$, let $\Delta t := T / N_T$ be the time-step where $T$ is the time corresponding to the end of the simulation. Let $t_n := n \Delta t$, $n = 0, \dots, N_T - 1$ be the temporal mesh. 
We approximate the continuous time derivative by $\frac{\partial u_h}{\partial t} \approx \frac{u^{n+1}_h - u^n_h}{\Delta t}$. We define
\begin{equation*}
    u_h^n(x) := \sum_{j =1}^{N_h} u^n_j \chi_j(x)\ , \quad \text{and} \quad c_h^n(x) := \sum_{j = 1}^{N_h} c^n_j \chi_j(x)\ ,
\end{equation*}
the finite element approximations of the cell density $u$ and the concentration of chemoattractant $c$
where $\{u^n_j\}_{j=1,\dots,N_h}$ and $\{c^n_j\}_{j=1,\dots,N_h}$ are unknowns and $\{\chi_j\}_{j=1,\dots,N_h}$ is the finite element basis. Then, the finite element problem associated with the system \eqref{eq: final system 1} reads as follows.

For each $n = 0, \dots, N_T -1$, find~$\{u_h^{n+1},c_h^{n+1}\}$ in $S^h \times S^h$, such that for all $\chi \in S^h$
\begin{align}
\left(\frac{u_h^{n+1}-u_h^n}{\Delta t},\chi\right)^h &+ \left({D}(u_h^n)\nabla u_h^{n+1},\nabla\chi\right) = \nonumber\\ & \left(A\left(\phi^{\text{upw}}(u_h^n)\right)\nabla c_h^n,\nabla \chi \right) + r_0 \left(u_h^n \left(1-\frac{u_h^n}{u_{\textnormal{max}}}\right),\chi\right)^h ,\label{eq:discr-eq-11}\\
\zeta \left(\frac{c_h^{n+1}-c_h^n}{\Delta t},\chi\right)^h &= -\left(\nabla c^{n+1}_h,\nabla \chi\right)+ \left(u_h^{n+1}-c_h^{n+1},\chi \right)^h.
\label{eq:discr-eq-12}
\end{align}

The finite element scheme associated with the system \eqref{eq: nondimensional second phase} including the effect of the treatment is the following
\begin{align}
\theta \Bigl(\frac{w_h^{n+1}-w_h^n}{\Delta t}&,\chi\Bigr)^h  + \left(\overline{D}(w_h^n,M^n_h)\nabla w_h^{n+1},\nabla\chi\right)\nonumber \\
&=\left(B\left(\overline{\phi}^{\text{upw}}_2(w_h^n,M^n_h)\right)\nabla c_h^n,\nabla \chi \right) + \tilde{r}\left(w^n_h\Bigl(1-\frac{w^n_h}{u_{\textnormal{max}}} \Bigr),\chi\right)^h,\label{eq:discr-eq-21} 
\end{align}
\begin{align}
\zeta \left(\frac{c_h^{n+1}-c_h^n}{\Delta t},\chi\right)^h &= -\left(\nabla c^{n+1}_h,\nabla \chi\right)+ \left(w_h^{n+1}-c_h^{n+1},\chi \right)^h,\\
m \left(\frac{M_h^{n+1}-M_h^n}{\Delta t},\chi\right)^h &= -\left(\nabla M^{n+1}_h,\nabla \chi\right)-\delta \left(M_h^{n+1},\chi \right)^h. \label{eq:discr-eq-23}
\end{align}

\par 

In order to describe how the chemotactic coefficients $\phi^{\text{uwp}}$ and $\overline{\phi}^{\text{upw}}_2$ are computed, let us rewrite the discrete equation \eqref{eq:discr-eq-11} into its matrix form
\[
(G_L  + \Delta t K_D ) \underline{u}^{n+1} = G_L \underline{u}^{n} + \Delta t K_\phi \underline{c}^{n} + \Delta t G_L \underline{g}^n\ ,
\]
where $\underline{u}^n$ and $\underline{c}^n$ are the vectors of coefficients which are the unknowns of the problem and $\underline{g}^n$ is a vector defined by  
\[
\left[\underline{g}^n\right]_i= \left(u_h^n \left(1-\frac{u_h^n}{u_{\textnormal{max}}}\right)\right)(x_i)\ ,\ \quad \text{for } i=1,\dots,N_h\ . 
\]
We define the finite element matrices associated with the diffusion $K_D$ and the advection $K_\phi$
\begin{equation}
K_{D,ij} = \int_{\rm{\Omega}} D(u^n_h) \nabla \chi_i \nabla \chi_j \, \diff x \quad \textnormal{ for } i,j = 1, \dots, N_h\ ,
\label{eq:matrix-diff}
\end{equation}
\begin{equation}
K_{\phi,ij} = \int_{\rm{\Omega}} \phi^{\text{upw}}\left(u^n_h(x_i),u^n_h(x_j)\right) \nabla \chi_i \nabla \chi_j \, \diff x  \quad  \textnormal{ for } i,j = 1, \dots, N_h\ .
\label{eq:matrix-adv}
\end{equation}
In \eqref{eq:matrix-diff}, the integral is computed using Gauss quadrature to deal with a potential choice of nonlinear functional for $D(u^n_h)$. The exactness of the quadrature is obtained using the adequate number of Gauss points since $D(u^n_h)$ is a polynomial of order $\gamma(M)+1$.  

  The chemotactic coefficient is computed using an upwing approach. For each element and depending on the direction of the gradient of the chemoattractant we have 
\begin{equation}
\phi^{\text{upw}}\left(u^n_h(x_i),u^n_h(x_j)\right) = \begin{cases}u^n_h(x_j)\left(1-\left(\frac{u^n_h(x_i)}{\overline{u}}\right) \right), \quad &\text{if} \quad c^{n}_{h}(x_j) - c^{n}_{h}(x_i) < 0,\\
u^n_h(x_i)\left(1-\left(\frac{u^n_h(x_j)}{\overline{u}}\right) \right), \quad &\text{otherwise}.
\end{cases}\label{eq: phi numerics}
\end{equation} 
Therefore, the chemotactic coefficient is chosen as function of the sign of the difference of chemoattractant between nodes connected by an edge. The same method is applied to compute $\overline{\phi}^{\text{upw}}_2$ in \eqref{eq:discr-eq-21}. 
The property of non-negativity of the cell density satisfied by our numerical scheme can be proved using similar arguments as in \cite{bubba2019positivity}.

\section{One dimensional numerical results}
\label{app: one dimension numerics}

\noindent\textbf{Influence of the domain size}

\noindent The unstable wavenumbers are discrete values, $k=n\pi/L$, that satisfy the relation (\ref{eq: unstable modes}) from Section \ref{sec: first phase}. The wavemode $n$ determines the number of aggregates depending on the length of the domain. For $A=7$ and the parameters specified at the beginning of this section we have $0.25<n\pi/L<0.4$. As shown in Figure \ref{fig: different L}, as we increase the length of the domain, the number of aggregates also increases. When the domain is large, as in Figure \ref{fig: L100}, we observe that some aggregates are merging together while others are emerging, \emph{i.e.}, they are formed from a zone of low cell density. This process is called coarsening \cite{painter2002volume} and is not observed in a small domain such as in Figure \ref{fig: L10}.

\begin{figure}[tbhp]
  \centering
  \subfloat[]{\label{fig: L10}\includegraphics[scale=0.25]{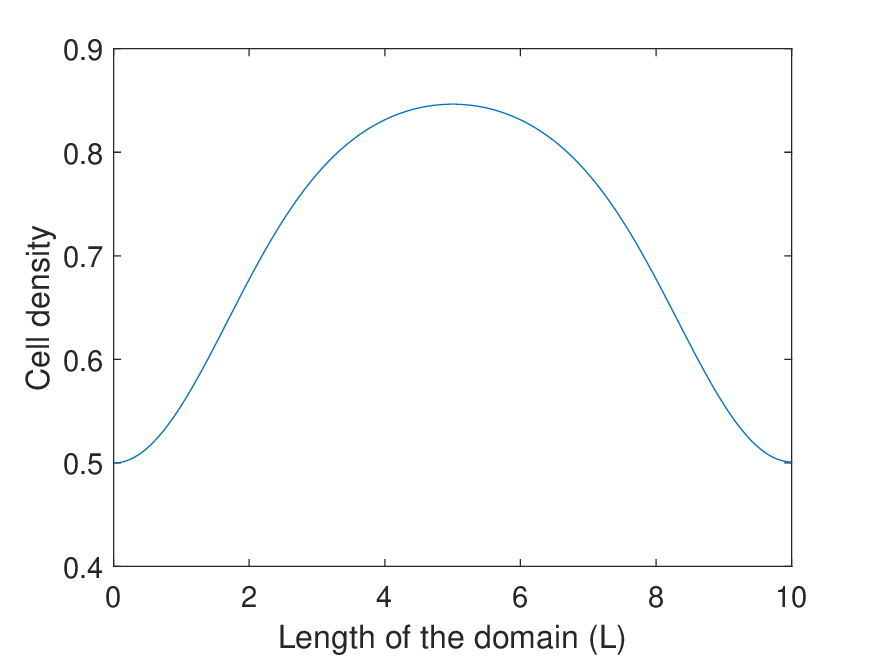}}
  \subfloat[]{\label{fig: L40}\includegraphics[scale=0.25]{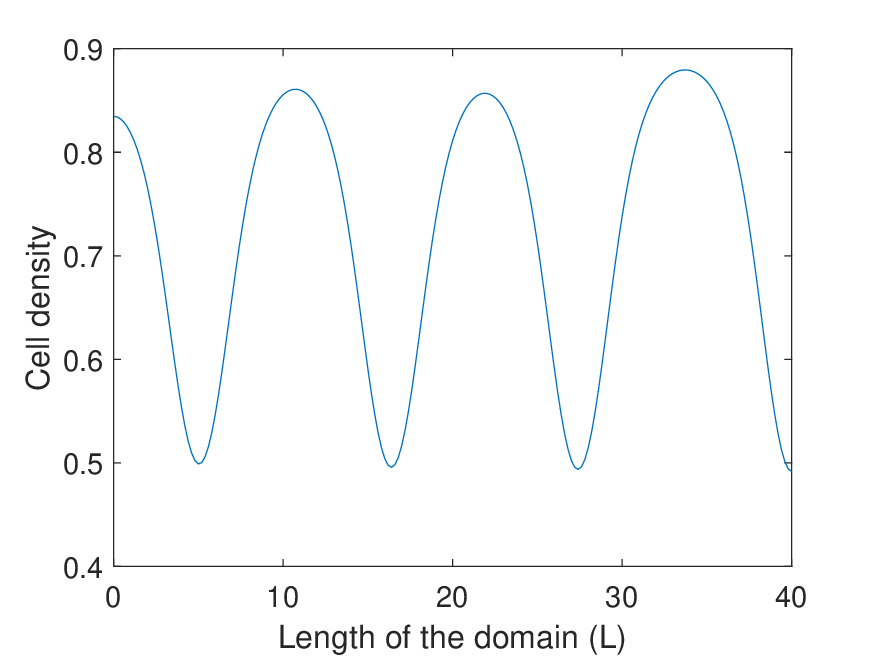}}
  \subfloat[]{\label{fig: L100}\includegraphics[scale=0.25]{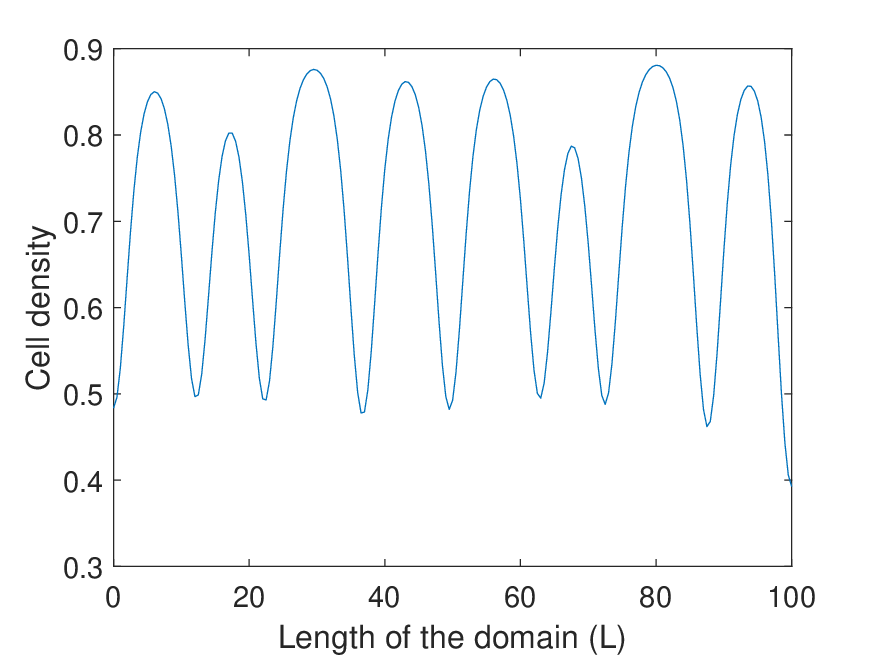}}
  \caption{Relationship between the wavenumber $k$ and the length of the domain $L$ for $A=7$.} 
  \label{fig: different L}
\end{figure}

For the remaining simulations in this section we fix the length of the domain to $L=40$, and all simulations are performed with carrying capacity $u_{\max} = 0.5$. \\

\noindent\textbf{Comparison of the results with the stability analysis predictions}
{Here, we study the influence of the chemosensitivity parameter $A$ on the pattern dynamics and size of the aggregates, in the presence or absence of TMZ. In order to compare the solutions to the predictions of the stability analysis, the initial condition is a small perturbation around the homogeneous distribution $u_0 = 0.5$. Results of this section are obtained with a proliferation rate $r_0=0.1$. For such parameters, using the results of Section \ref{sec: linear stability}, the critical value of the chemosensitivity parameter without the treatment (in \textbf{P1}, when $M=0$ and therefore $\gamma(M) = 1$ ) is $A^c \approx 6.92$ and with the treatment uniformly distributed (for $\gamma(M) = 5$) is $B^c \approx 3.9$.}

{In the first part of the experiment, \emph{i.e.} without any treatment, we show in Figure \ref{fig: formation spheroids 01} the formation of patterns at different times, for $A = 7$ (close to the instability threshold, Figure \ref{fig: formation spheroids 7}), $A = 50 $ (Figure \ref{fig: formation spheroids 50}) and $A=150$ (Figure \ref{fig: formation spheroids 150}). We observe here the process of merging and emerging patterns through time.}  

\begin{figure}[tbhp]
  \centering
  \hspace{-0.8cm}\subfloat[\label{fig: formation spheroids 7}]{\includegraphics[scale=0.4]{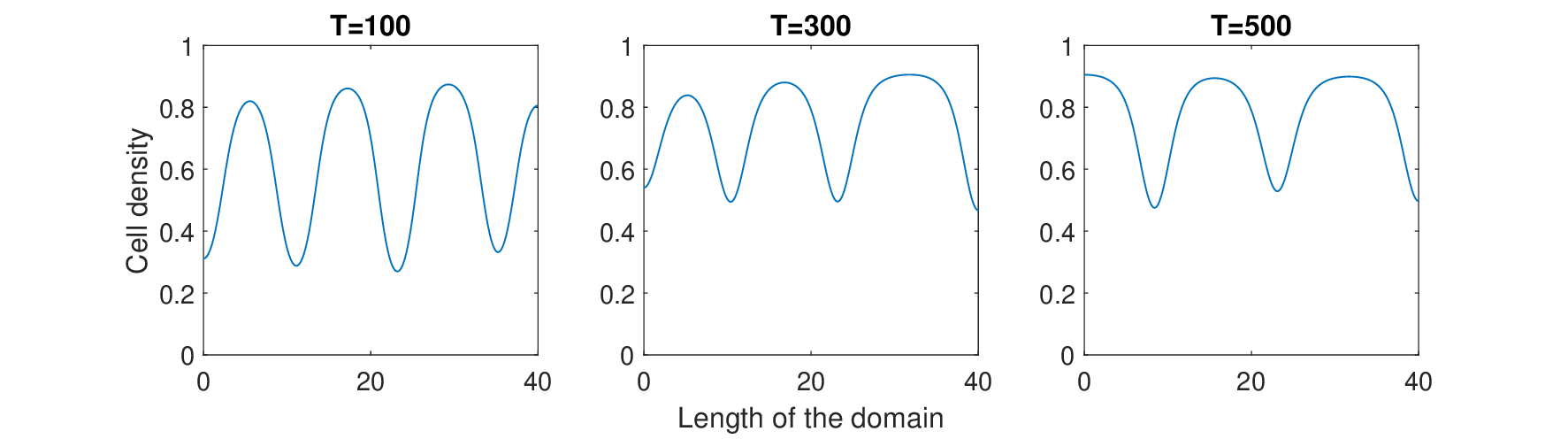}}
  
\hspace{-0.8cm}\subfloat[\label{fig: formation spheroids 50}]{
 \includegraphics[scale=0.4]{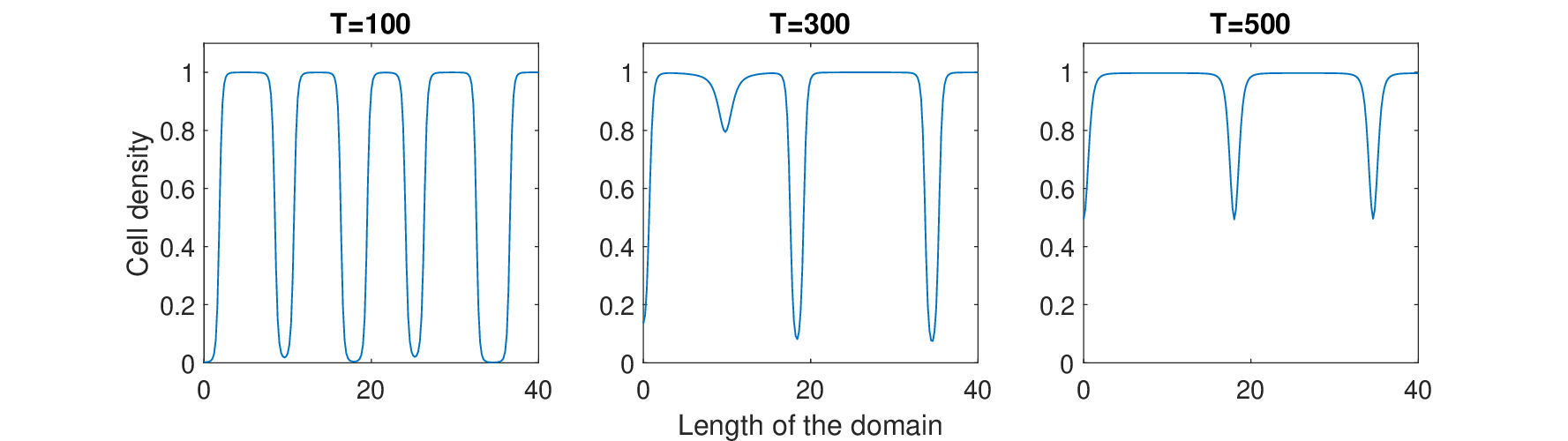}}

\hspace{-0.6cm}\subfloat[\label{fig: formation spheroids 150}]{\includegraphics[scale=0.4]{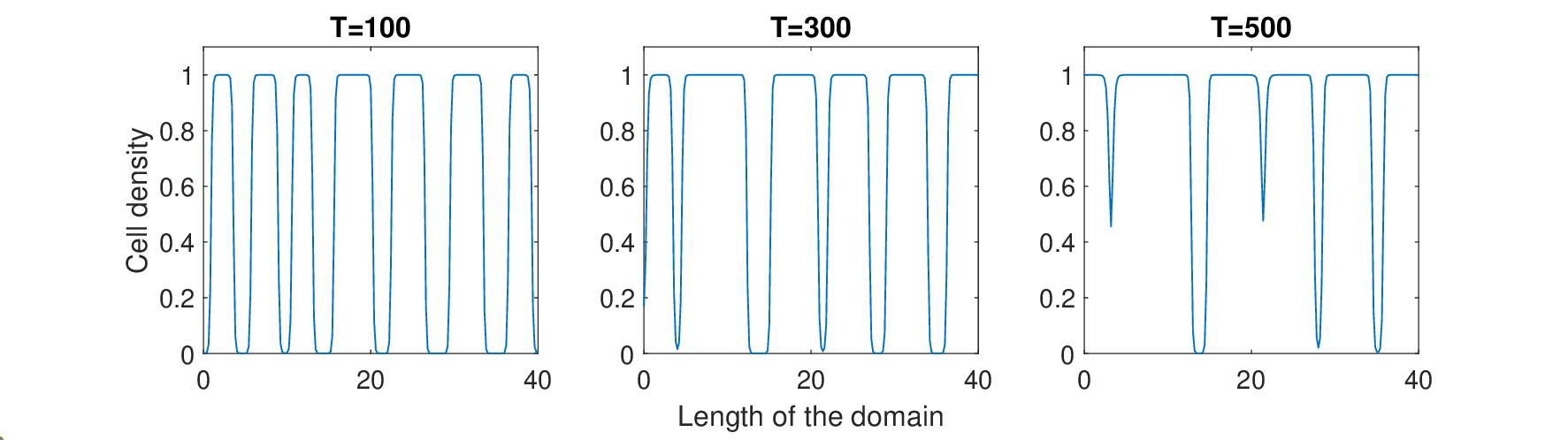}}
    \caption{Formation of the pattern without the treatment for (a) $A=7$, (b) $A=50$ and (c) $A=150$ when $r_0=0.1$ and $u_{\textnormal{max}}=u_0=0.5$. }\label{fig: formation spheroids 01}
\end{figure}

{As predicted by the stability analysis, larger values of the chemotactic sensitivity $A$ favor the emergence of smaller aggregates (compare Figures \ref{fig: formation spheroids 7} and \ref{fig: formation spheroids 150} for $A = 7$ and $A=150$, respectively). This is due to the fact that for larger chemotactic sensitivity, cells are more attracted to zones of high concentration of chemoattractant. The creation of patterns instead of the expansion of a homogeneous cell distribution is due to an instability which results from a positive feedback loop between the production of the chemical by the cells on one hand, and their attraction to high density zones of this chemical on the other. The chemotactic sensitivity $A$ must be large enough to trigger this instability, in order to compensate the competing effects of diffusion and of the logistic growth term, which on the contrary, tends to regulate the local cell density to the carrying capacity of the environment $u_\textnormal{max}$, and therefore induces cell death inside the aggregated patterns for which the density is above $u_\textnormal{max}$.}

When the drug is introduced uniformly in the domain starting from a homogeneous distribution of cells (\emph{i.e.} for $M=1$, $\gamma(M) = 5$) at time $t=0$, we show in Figure \ref{fig: formation spheroids 01 TMZ} the formation of patterns at different times, for chemosensitivity $B = 5$ (close to the instability threshold, Figure \ref{fig: formation spheroids 5 TMZ}), $B = 30 $ (Figure \ref{fig: formation spheroids 30 TMZ}) and $B=150$ (Figure \ref{fig: formation spheroids 150 TMZ}). Note that in this case we also let cells to proliferate with rate $r_0=0.1$.

\begin{figure}[tbhp]
  \centering
  \hspace{-0.8cm}\subfloat[\label{fig: formation spheroids 5 TMZ}]{\includegraphics[scale=0.4]{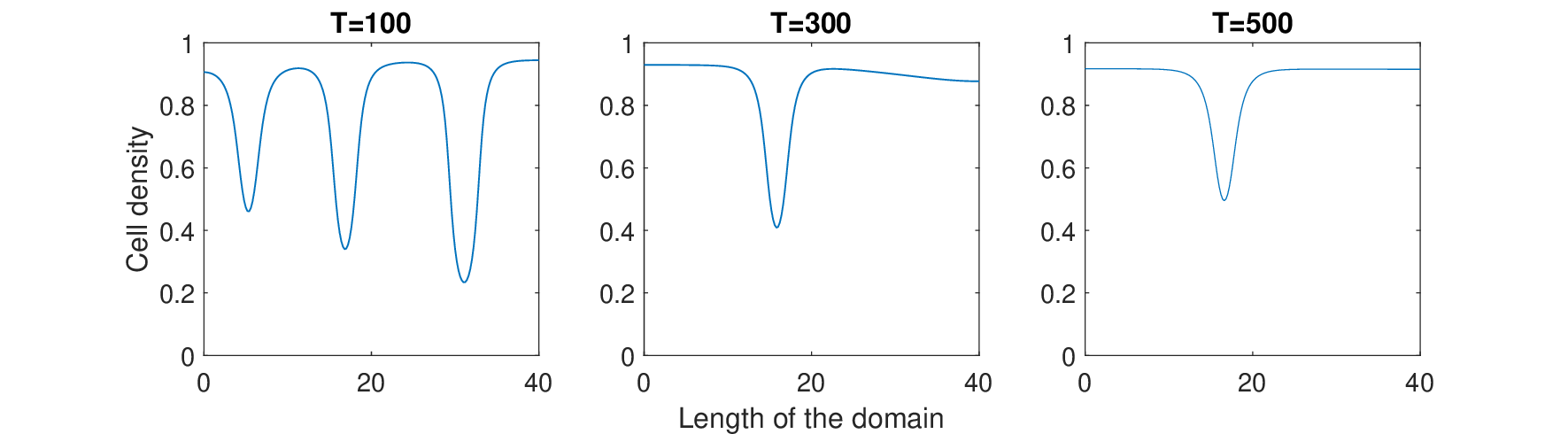}}
  
\hspace{-0.8cm}\subfloat[\label{fig: formation spheroids 30 TMZ}]{
 \includegraphics[scale=0.4]{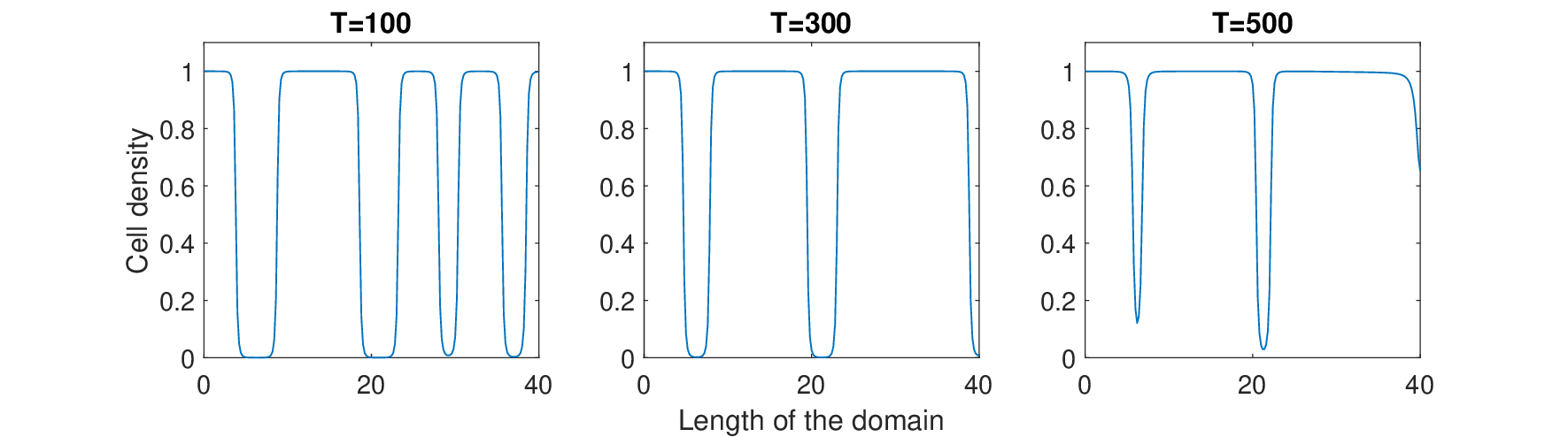}}

\hspace{-0.6cm}\subfloat[\label{fig: formation spheroids 150 TMZ}]{\includegraphics[scale=0.4]{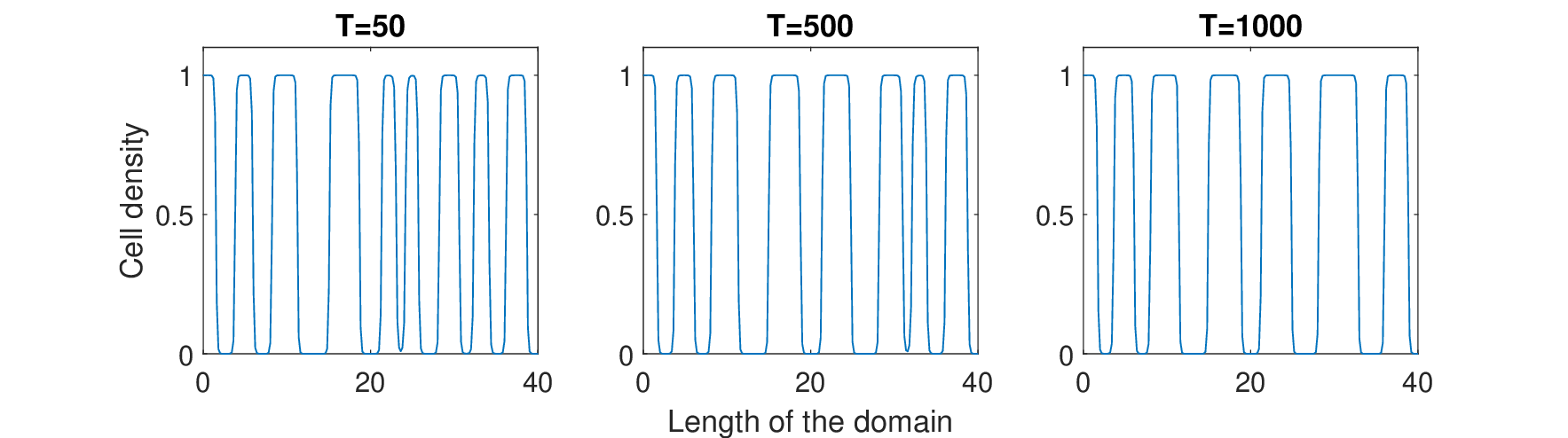}}
    \caption{Formation of the pattern with treatment ($M=1, \gamma(M) = 5$), included at $t=0$ when $u_0=0.5$, $r_0=0.1$ for (a) $B=5$, (b) $B=30$ and (c) $B=150$. }\label{fig: formation spheroids 01 TMZ}
\end{figure}

{As one can see in Figure \ref{fig: formation spheroids 01 TMZ}, we first observe again that increasing the chemosensitivity parameter $B$ results in the formation of smaller cell aggregates (compare Figures \ref{fig: formation spheroids 5 TMZ} and \ref{fig: formation spheroids 30 TMZ}). Very close to the instability threshold (Figure \ref{fig: formation spheroids 5 TMZ}), the system converges quickly to one aggregate, while for larger values of $B$ (Figure \ref{fig: formation spheroids 150 TMZ}) a large number of well-separated small aggregates arises. These clusters merge in time to form bigger clusters as for the case without the treatment. Moreover, comparing Figures. \ref{fig: formation spheroids 150} and \ref{fig: formation spheroids 150 TMZ}, we clearly observe that varying the mechanical state of cells (\emph{i.e.} passing from $\gamma(M) = 1$ to $\gamma(M) = 5$), leads to a change in the cell's aggregate size. When cells are more elastic, they tend to create smaller aggregates than when they behave as rigid spheres. }



\section*{Acknowledgments}

The authors would like to thank K.~J.~Painter and H.~Gimperlein for fruitful discussions and helpful suggestions. This research was funded by the Plan
Cancer/ITMO HTE (tumor heterogeneity and ecosystem) program (MoGlimaging).


%
%



\bibliographystyle{abbrv}
\bibliography{Glioblastoma_ref}

\end{document}